\documentclass[11pt, letterpaper]{article}

\usepackage{algorithm}
\usepackage[noend]{algpseudocode}

\usepackage[utf8]{inputenc}
\usepackage[T1]{fontenc}

\usepackage{cite}
\usepackage[pdftex]{graphicx}
\usepackage{subcaption}
\usepackage{threeparttable}
\usepackage[margin=1in]{geometry}
\usepackage{amsmath,setspace,amssymb}
\usepackage{amsthm}
\usepackage{comment}
\usepackage{enumerate}
\usepackage[unicode,pdfencoding=auto]{hyperref}
\usepackage{bookmark}
\usepackage{color} 
\usepackage{thm-restate}

\graphicspath{{./}}

\theoremstyle{plain}
\newtheorem{theorem}{Theorem}[section]
\newtheorem{lemma}[theorem]{Lemma}
\newtheorem{corollary}[theorem]{Corollary}

\theoremstyle{definition}
\newtheorem{definition}[theorem]{Definition}

\newcommand{\Acal}{\mathcal{A}}

\newcommand{\Ccal}{\mathcal{C}}

\newcommand{\Pcal}{\mathcal{P}}
\newcommand{\Qcal}{\mathcal{Q}}

\newcommand{\Scal}{\mathcal{S}}

\newcommand{\Dist}{\mathsf{dist}}
\newcommand{\Deg}{\mathsf{deg}}

\allowdisplaybreaks[2]

\title{Beating Quadratic Time--Message Trade-off \\ in Distributed Minimum Spanning Tree Construction}
\author{Taisuke Izumi\thanks{The University of Osaka. Emails: izumi.taisuke.ist@osaka-u.ac.jp, \{n-kitamura, masuzawa\}@ist.osaka-u.ac.jp}
\and Naoki Kitamura$^*$
\and Toshimitsu Masuzawa$^*$
}

\date{}

\begin{document}

\maketitle
\thispagestyle{empty}

\begin{abstract}
  \sloppy{
    We present a new distributed algorithm for computing a minimum spanning tree (MST) in the
    \textsf{CONGEST-KT$_{1}$} model, where messages are limited to $O(\log n)$ bits and each vertex initially
    knows the identifiers of its neighbors. Our algorithm exposes a two-parameter time--message trade-off:
    for any $0 \leq \lambda \leq \kappa \leq 1/2$, it runs in
    $\tilde{O}(n^{\lambda}D_G + n^{1 - \kappa - \lambda} + n^{1 - 2\kappa + \lambda} + n^{1/2})$ rounds and uses
    $\tilde{O}(\min\{m, n^{1 + \kappa}\})$ messages, where $n$, $m$, and $D_G$ are the number of vertices, edges, and the
    network diameter, respectively. In particular, setting $(\kappa, \lambda) = (1/3, 1/6)$ yields an MST algorithm
    running in $\tilde{O}(n^{1/2} + n^{1/6}D_G)$ rounds with only $\tilde{O}(n^{4/3})$ messages. Under the mild
    assumption $D_G = O(n^{1/3})$, this is round-optimal while improving the best known message bound of
    $\tilde{O}(n^{3/2})$.
    More broadly, our algorithm breaks the quadratic time--message trade-off barrier
    $\mathrm{\# rounds} \cdot \mathrm{\# messages} = \tilde{\Omega}(n^2)$, which
    no previous MST algorithm in the \textsf{CONGEST-KT$_{1}$} model has been able to overcome, and it does so for almost the
    entire range of the diameter $D_G$. As a byproduct, we also obtain new low-message broadcast,
    spanning-tree, and leader-election algorithms.

    The key technical idea behind our algorithm is a new message-efficient construction of $(\alpha, \beta)$-spanners:
    for any $0 \leq \lambda \leq \kappa \leq 1/2$, we build an $(\alpha, \beta)$-spanner with
    $\alpha = \tilde{O}(n^{\lambda})$ and $\beta = \tilde{O}(n^{1 - \kappa - \lambda} + n^{1 - 2\kappa + \lambda})$
    using $\tilde{O}(n^{1 + \kappa})$ messages. Our construction combines the low-diameter decomposition of
    Miller, Peng, and Xu [SPAA, 2013] with two ideas newly introduced in this paper: an $(\alpha, \beta)$-partition
    framework that recasts spanner construction as a partitioning task, and a novel \emph{degree-adaptive clustering}
    technique. To achieve the desired message complexity, we further develop a non-trivial analysis that disentangles the intricate correlation between the random edge sampling and the low-diameter decomposition computed over the sampled edges. This analysis itself is also of independent interest.

  }
\end{abstract}

\setcounter{page}{0}
\newpage
\section{Introduction}

\subsection{Background}
\label{subsec:background}

The efficiency of distributed graph algorithms is traditionally measured by
two fundamental resources, \emph{round complexity} and \emph{message complexity}.
Over the past three decades, the main focus of research in the \textsf{CONGEST}
model has been on minimizing round complexity. For many global graph problems,
including minimum spanning tree (MST), minimum cut, shortest paths, subgraph connectivity,
and a variety of verification problems, the optimal round complexity is now known to be
$\tilde{\Theta}(D_G + \sqrt{n})$, where $n$ and $D_G$ respectively denote the number of vertices and the diameter of the network~\cite{PR00,SHKKNPPW12}.
In contrast, the landscape of \emph{message complexity}
is more subtle and depends critically on the initial knowledge assumptions of the model,
which are often captured by the notion of \textsf{KT$_\rho$} (Knowledge Till radius $\rho$).
The \textsf{KT$_\rho$} assumption means that each vertex $v$ in the network initially knows
the information in the radius-$\rho$ ball centered at $v$. Two major settings are
$\rho = 0, 1$, that is, each vertex has no initial knowledge, or knows the identifiers of
its neighbors. It has been known that these assumptions exhibit a clear separation
in message complexity for many global problems such as broadcast, leader election, spanning tree,
and MST. Under the \textsf{KT$_0$} assumption, those problems inherently require $\tilde{\Omega}(m)$
messages\footnote{While the original paper states that this lower bound holds only
  for deterministic algorithms or comparison-based randomized algorithms (without explicit
  assumption of \textsf{KT$_0$}),
  this lower bound holds for any algorithm under the \textsf{KT$_0$} assumption~\cite{KPPRT15}.},
where $m$ is the number of edges in the input graph~\cite{AGPV90}.
On the other hand, under the \textsf{KT$_1$} assumption, the
additional knowledge enables fundamentally different algorithmic techniques and allows us
to circumvent the classical $\tilde{\Omega}(m)$ message lower bound in \textsf{KT$_0$}.
A breakthrough result by King, Kutten, and Thorup~\cite{KKT15}
demonstrated that an MST can be constructed using only
$\tilde{O}(n)$ messages in the \textsf{CONGEST-KT$_1$} model.
However, it requires $\tilde{O}(n)$ rounds,
which is far from the near-optimal $\tilde{O}(D_G + \sqrt{n})$ round complexity.

Building on this insight,
Gmyr and Pandurangan~\cite{GP18} and Ghaffari and Kuhn~\cite{GK18}
independently studied time--message trade-offs in the \textsf{CONGEST-KT$_1$} model.
They presented a randomized algorithm for constructing a spanning subgraph of
diameter $\tilde{O}(D_G + n^{1-c})$ using $\tilde{O}(\min\{m, n^{1+c}\})$ messages.
This immediately yields algorithms for broadcast, leader election, and spanning tree that run
in $\tilde{O}(n^{1 - c} + D_G)$ rounds and use $\tilde{O}(\min\{m, n^{1 + c}\})$ messages
for a given parameter $0 \leq c \leq 1$.
For harder global problems, including MST and minimum-cut approximation,
they also presented algorithms attaining near-optimal $\tilde{O}(\sqrt{n} + D_G)$ round complexity and
$\tilde{O}(\min\{m, n^{3/2}\})$ message complexity by plugging in the techniques of \cite{KKT15}.

\subsection{Our Contribution}

This paper also investigates the message complexity in the \textsf{CONGEST-KT$_1$} model.
Our central question is whether the worst-case trade-off between message complexity and round complexity
\begin{align*}
  \text{\#messages} \cdot \text{\#rounds}
  = \tilde{\Omega}(n^2)
\end{align*}
constitutes an inherent barrier for global distributed computation.
This quadratic product bound has implicitly guided the design of distributed
algorithms for decades: all previously known near-optimal algorithms conform
to it, and no prior work has succeeded in breaking this trade-off in the
\textsf{CONGEST-KT$_1$} model. Our main result demonstrates that this barrier
is \emph{not} fundamental. The key technical result is a message-efficient construction of
an $(\alpha, \beta)$-spanner, that is, a spanning subgraph $H$ of $G$ such that any pair
$u, v \in V(G)$ admits a path of length at most $\alpha \cdot \Dist_G(u, v) + \beta$ in $H$, where
$\Dist_G(u, v)$ is the $u$--$v$ distance in $G$.

\begin{theorem} \label{thm:main}
  Let $0 \leq \lambda \leq \kappa \leq 1/2$ be arbitrary parameters. There exists
  a randomized algorithm for constructing an $(\alpha, \beta)$-spanner in the \textsf{CONGEST-KT$_1$} model, which
  uses $\tilde{O}(\min\{m, n^{1 + \kappa}\})$ messages, runs in $\tilde{O}(\alpha + \beta + n^{\kappa})$ rounds, and attains $\alpha = \tilde{O}(n^{\lambda})$ and $\beta = \tilde{O}(n^{1 - \kappa - \lambda} + n^{1 - 2\kappa + \lambda})$
  with probability $1 - o(1)$.
\end{theorem}

To obtain the MST from this spanner structure, one can use the technique of the message-efficient partwise aggregation by Haeupler, Hershkowitz,
and Wajc~\cite{HHW18} combined with the approach by Gmyr and Pandurangan~\cite{GP18} and Ghaffari and Kuhn~\cite{GK18}, which provides the following message-efficient algorithm.

\begin{corollary} \label{corol:MST}
  Let $0 \leq \lambda \leq \kappa \leq 1/2$ be arbitrary parameters. There exists a randomized \textsf{CONGEST-KT$_1$} algorithm for
  MST,
  which runs in $\tilde{O}(n^{\lambda}D_G +
    n^{1 - \kappa - \lambda} + n^{1 - 2\kappa + \lambda} + n^{1/2})$ rounds and uses $\tilde{O}(\min\{m, n^{1 + \kappa}\})$ messages. It succeeds with probability $1 - o(1)$.
\end{corollary}

Focusing on the regime of (near) round-optimality with respect to $n$ (i.e., $\tilde{\Theta}(n^{1/2})$ rounds), the setting $(\kappa, \lambda) = (1/3, 1/6)$ yields round complexity $\tilde{O}(n^{1/2})$ and message complexity $\tilde{O}(n^{4/3})$ for $D_G = O(n^{1/3})$. More generally, supposing $D_G \leq n^{\delta}$ for a known constant $\delta$ with $\delta = 1/2 - \Omega(1)$, the setting $(\kappa, \lambda) = (\delta,\, 1/2 - \delta)$ yields round complexity $\tilde{O}(n^{1/2})$ and message complexity $\tilde{O}(n^{1 + \delta}) = \tilde{O}(n^{3/2 - \Omega(1)})$. In either case our algorithm is round-optimal while improving the best known message bound of $\tilde{O}(n^{3/2})$ by a polynomial factor.

Moving away from $\tilde{\Theta}(n^{1/2})$ dependence on $n$, Corollary~\ref{corol:MST} breaks the quadratic time--message barrier for almost the entire range of the diameter $D_G$.
The case $D_G = \Omega(n)$ is trivially excluded, since the $\Omega(D_G)$-round and $\Omega(n)$-message lower bounds
already force $\text{\#rounds} \cdot \text{\#messages} = \Omega(n^2)$. We therefore focus on $1/3 < \delta = 1 - \Omega(1)$ (the complementary regime $\delta \leq 1/3$ is already handled round-optimally above, with the best possible message complexity $\tilde{O}(n^{4/3})$). Then setting $(\kappa, \lambda) =
  \left(\frac{1-\delta}{2}, \frac{1-\delta}{4}\right)$ balances the first three terms of the round complexity at the common value $n^{(1+3\delta)/4}$, which exceeds $n^{1/2}$ since $\delta > 1/3$. Hence the round complexity is $\tilde{O}(n^{(1+3\delta)/4})$ and the
message complexity is $\tilde{O}(n^{(3-\delta)/2})$, yielding the product $\tilde{O}(n^{(7+\delta)/4}) = \tilde{O}(n^{2 - (1-\delta)/4})$, a polynomial factor below $n^2$.

Theorem~\ref{thm:main} also yields message-efficient broadcast, spanning tree construction, and leader
election algorithms: it suffices to construct the $(\alpha, \beta)$-spanner and then broadcast over it,
which takes $\tilde{O}(\alpha D_G + \beta)$ additional rounds.

\begin{corollary} \label{corol:bcast}
  Let $0 \leq \lambda \leq \kappa \leq 1/2$ be arbitrary parameters. There exist randomized
  \textsf{CONGEST-KT$_1$} algorithms for broadcast, spanning tree, and leader election that run in
  $\tilde{O}(n^{\lambda}D_G + n^{1 - \kappa - \lambda} + n^{1 - 2\kappa + \lambda} + n^{\kappa})$
  rounds and use $\tilde{O}(\min\{m, n^{1 + \kappa}\})$ messages. They succeed with probability $1 - o(1)$.
\end{corollary}
For $\lambda = \kappa/2$, $\kappa \leq 1/3$, and $D_G = O(n^{1 - 2\kappa})$,
the corollary gives $\tilde{O}(n^{1 - 3\kappa/2})$ rounds with $\tilde{O}(\min\{m, n^{1 + \kappa}\})$ messages. This improves upon the corresponding algorithms of~\cite{GP18,GK18},
which run in $\tilde{O}(n^{1 - \kappa})$ rounds under the same message bound. In addition, just as for MST (Corollary~\ref{corol:MST}), these algorithms break the quadratic barrier over almost the entire range of the diameter with the same complexity as the MST.

Although attaining a \emph{singularly optimal} MST algorithm, one that is simultaneously optimal in both round and message complexity, still remains an intriguing open problem in the \textsf{CONGEST-KT$_1$} model, our result marks a significant step forward by breaking a barrier that has constrained all prior algorithms.

\subsection{Technical Outline}
\label{subsec:outline}

\paragraph{The $(\alpha, \beta)$-Partition Framework}
To highlight the novelty of our algorithm, we first give an overview of the approach
commonly used in prior work~\cite{GP18,GK18}, attaining the round complexity
$\tilde{O}(n^{1/2} + D_G)$ and the message complexity $\tilde{O}(n^{3/2})$. The central idea
in these results is to construct a $(\tilde{O}(1), \tilde{O}(\sqrt{n}))$-spanner with
$\tilde{O}(n^{3/2})$ edges using $\tilde{O}(n^{3/2})$ messages. At the core of this approach
lies a Baswana--Sen-type spanner construction~\cite{BS07}.

Vertices of degree at most $\sqrt{n}$ can afford to include all their incident edges
in the spanner, since the total number of such edges is bounded by
$\tilde{O}(n^{3/2})$.
Therefore, the main difficulty arises from vertices whose degree is at least
$\sqrt{n}$.
To handle these high-degree vertices, prior works sample each vertex
independently with probability $\tilde{\Theta}(1/\sqrt{n})$ and use the sampled
vertices as cluster centers.
High-degree vertices are then clustered around these centers.
To maintain global connectivity, outgoing edges of clusters are detected following
the technique of graph sketches~\cite{AGM12}, and cluster merging by detected edges
is repeatedly performed. As a result of these merging steps, one obtains a spanning
subgraph of diameter $\tilde{O}(\sqrt{n})$, which forms a $(\tilde{O}(1), \tilde{O}(\sqrt{n}))$-spanner.

To generalize this construction, we first introduce the notion of an
\emph{$(\alpha, \beta)$-partition} (Definition~\ref{dfn:abc-partition}).
An $(\alpha, \beta)$-partition of a graph $H$ is a pair $(\Pcal, F)$
of a vertex partition $\Pcal$ and an edge set $F \subseteq E(H)$
such that (1) every part $P \in \Pcal$ has diameter at most $\alpha$ in the subgraph restricted to $F$-edges,
and (2) the sum of the diameters over all parts that have boundary edges not captured by $F$
(called \emph{open} parts) is bounded by $\beta$.
In this context, the prior approach corresponds to constructing a $(\tilde{O}(1),
  \tilde{O}(\sqrt{n}))$-partition.
Each low-degree vertex is a closed part consisting of a single vertex, and
the clusters formed by the Baswana--Sen-type construction are regarded as open parts of diameter $O(1)$,
whose diameter sum is $\tilde{O}(\sqrt{n})$. Then, iterative merging of open parts can be seen
as the process of transforming a given $(\tilde{O}(1), \tilde{O}(\sqrt{n}))$-partition to a $(\tilde{O}(1),
  \tilde{O}(\sqrt{n}))$-spanner. Importantly, this process is easily generalized to an arbitrary
$(\alpha, \beta)$-partition (Lemma~\ref{lma:spanner-from-partition}). That is, an $(\alpha,\beta)$-partition
$(\Pcal, F)$ immediately yields an $(O(\alpha),O(\beta))$-spanner of $|F| + \tilde{O}(n)$ edges via a procedure that
consumes only $\tilde{O}(n)$ additional messages. Hence the algorithmic challenge reduces to constructing an
$(\alpha,\beta)$-partition efficiently.

\paragraph{Starting Point: Low-Diameter Decomposition}
Let $0 \leq \lambda \leq \kappa \leq 1/2$. Our target is a $(\tilde{O}(n^{\lambda}),\, \tilde{O}(n^{1-\kappa-\lambda} + n^{1-2\kappa+\lambda}))$-partition
constructed with only $\tilde{O}(n^{1+\kappa})$ messages. By the reduction above this yields Theorem~\ref{thm:main}. The key idea of our algorithm is to use the $d$-low-diameter
decomposition ($d$-LDD) algorithm of Miller, Peng, and Xu~\cite{MPX13}, called $\textsf{MPX}(d)$ hereafter. It partitions
the input graph into connected components of diameter $O(d \log n)$ with $O(n/d)$ boundary vertices
(in expectation)~\footnote{
  In the original proof of \textsf{MPX}, the boundary is defined as a set of edges, but not vertices (i.e.,
  the probability that an edge crosses two different parts is $O(d^{-1})$). However, it is straightforward
  to modify that claim in the one considered here (see Appendix~\ref{appendix:MPX}).
}.
More precisely, the starting point of our algorithm is stated as follows:
Each vertex $u$ independently samples each of its incident edges with probability $n^{\kappa}/\Deg_G(u)$,
where $\Deg_G(u)$ is the degree of $u$ in $G$. Let $J$ denote the sampled edge set.
We then run $\textsf{MPX}(n^{\lambda})$ on the spanning subgraph $(V(G), J)$ and
obtain a partition $\Pcal$ of $V(G)$ such that each $P \in \Pcal$ induces a subgraph of
diameter $\tilde{O}(n^{\lambda})$.
If we can close all but $\tilde{O}(n^{1-\kappa-2\lambda} + n^{1-2\kappa})$ parts
by adding at most $\tilde{O}(n^{1+\kappa})$ edges to $J$, we obtain
the desired $(\tilde{O}(n^{\lambda}), \tilde{O}(n^{1-\kappa-\lambda} + n^{1-2\kappa+\lambda}))$-partition. However, making the parts closed turns out to be
highly non-trivial, mainly due to the following two reasons.
\begin{itemize}
  \item First, to make a part $P$ closed, one must include all the incident edges of
        its boundary vertices into the output edge set $J$.
        If a boundary vertex has high degree, this can contribute many edges to $J$,
        potentially causing the total size of $J$ to exceed the desired $\tilde{O}(n^{1+\kappa})$ message bound.
  \item Second, a part $P$ is required to be closed with respect to the edge set $E(G)$ of the \emph{original} graph $G$,
        not merely with respect to the sampled edge set $J$. That is, to close a part $P$, we must account for all
        edges incident to $P$ in $G$. Intuitively, one might expect the \emph{non-boundary} vertices of $P$ to emit only few non-sampled outgoing edges, for the following reason: by definition, a non-boundary vertex of $P$ has no outgoing edge, i.e., no edge of $J$ leaving $P$. Hence all of its outgoing edges in $G$ are non-sampled. Since each edge incident to a vertex $u$ is sampled into $J$ with probability $n^{\kappa}/\Deg_G(u)$, a vertex emitting many outgoing edges in $G$ would likely have had at least one of them sampled, thereby becoming a boundary vertex. One would therefore expect each non-boundary vertex of $P$ to emit only few outgoing edges in $G$ in the first place.

        This intuition is, however, incorrect in general,
        because the partition $\Pcal$ output by \textsf{MPX} and the sampled edge set $J$ are correlated:
        the algorithm uses $J$ to determine which vertices become boundary vertices, and therefore the event that a vertex subset
        $P \subseteq V(G)$ becomes a part and the event that $P$'s outgoing edges are sparse in $J$ are not independent.
\end{itemize}

\paragraph{Our Approach: Degree-Adaptive Clustering}
To overcome both challenges, we propose a novel technique called \emph{degree-adaptive clustering}, which is employed in the preprocessing step
of our algorithm.
For each degree class $\Delta_i = 2^i$, this clustering produces $\tilde{O}(n/\Delta_i)$ clusters of diameter
$O(\log n)$, each consisting of at most $O(\Delta_i)$ vertices whose degrees in $G$ are at most $\Delta_i$. Contracting each cluster to a supernode yields a multigraph $\widehat{G}$. Then, a vertex lying in a class-$i$ cluster samples each incident edge with probability $\Theta(n^{\kappa}\log n/\Delta_i)$. All sampled edges, which are added to the final output set $F$, induce a subgraph $\widetilde{G}$ of $\widehat{G}$. The key technical ingredient of our algorithm is that running \textsf{MPX} on $\widetilde{G}$ nicely resolves the aforementioned two challenges.

The clustering itself is simple. We process the $i$-th degree class (with $\Delta_i = 2^i$), which consists of the vertices of degree within $[2^i, 2^{i+1}]$, from the largest $i$ to the smallest. In processing the $i$-th class, each remaining vertex (i.e., a vertex not clustered yet) is independently sampled as a cluster center with probability $\tilde{\Theta}(1/\Delta_i)$, and every non-sampled vertex in the $i$-th class tries to attach itself to a cluster. If it has a sampled neighbor, it joins that center. Otherwise, it probes $O(\log n)$ of its randomly chosen neighbors and tries to find one that was already clustered in an earlier (higher-degree) iteration. If it is actually found, the vertex joins its cluster. Since a class-$i$ vertex has degree at least $\Delta_i$, one of these two cases always applies with high probability, and hence all such vertices are clustered. This produces $\tilde{O}(n/\Delta_i)$ clusters of diameter $O(\log n)$ per class.
The clusters produced this way may be much larger than $\Delta_i$, whereas our analysis requires each cluster to have size $O(\Delta_i)$. A post-processing step therefore splits every oversized cluster into subclusters of size $\Theta(\Delta_i)$, using the method of~\cite{IKNS22}. Although this splitting may cause different subclusters to overlap in a few shared vertices, one can handle it by replacing each shared vertex with several virtual copies, all simulated by the original vertex. The details are given in Section~\ref{subsec:balancing}.

\paragraph{Resolving the Two Challenges}
Employing the degree-adaptive clustering technique, our algorithm resolves the two challenges as follows.
\begin{itemize}
  \item A concern regarding the first challenge is that a boundary cluster in class $i$ is a supernode of degree $\tilde{O}(\Delta_i^2)$ in $\widehat{G}$
        (it consists of $\tilde{O}(\Delta_i)$ original vertices, each of degree at most $\Delta_i$), and thus a
        high-degree boundary supernode alone could contribute too many edges. We resolve this using the property
        of \textsf{MPX} that boundary vertices are spread out uniformly: when computing the $n^{\lambda}$-LDD, each
        supernode becomes a boundary vertex with probability $O(n^{-\lambda})$ (see Appendix~\ref{appendix:MPX}). Intuitively, our algorithm resolves the first challenge by combining this property with the fact that high-degree supernodes form only a small fraction of all supernodes in $\widetilde{G}$.

        First, let $\Ccal^H$ be the set of \emph{high-degree} supernodes, namely those in classes with $\Delta_i > n^{\kappa+\lambda}$.
        Since there are $\tilde{O}(n/n^{\kappa+\lambda}) = \tilde{O}(n^{1-\kappa-\lambda})$ supernodes in $\Ccal^{H}$, the expected number that become
        boundary vertices is $\tilde{O}(n^{1-\kappa-\lambda} \cdot n^{-\lambda}) = \tilde{O}(n^{1-\kappa-2\lambda})$. We therefore leave open
        every part that has a supernode in $\Ccal^H$ on its boundary: there are only $\tilde{O}(n^{1-\kappa-2\lambda})$ such parts, each of diameter
        $\tilde{O}(n^{\lambda})$. Hence they contribute only $\tilde{O}(n^{1-\kappa-2\lambda} \cdot n^{\lambda}) = \tilde{O}(n^{1-\kappa-\lambda})$ to the total open diameter.

        For every remaining part, all boundary supernodes lie in classes with $\Delta_i \leq n^{\kappa+\lambda}$. The number
        of edges added is then controlled \emph{not} by any per-part quantity but by the \emph{total} number of
        boundary supernodes in each class, which the same \textsf{MPX} property keeps small: the expected number of
        boundary supernodes in class $i$ is $\tilde{O}(n^{1 - \lambda}/\Delta_i)$
        (Lemma~\ref{lma:highsize}). Summing the $\tilde{O}(\Delta_i^2)$ incident edges over all
        boundary supernodes, class by class, gives
        \begin{align*}
          \sum_{i:\,\Delta_i \leq n^{\kappa+\lambda}} \tilde{O}\!\Bigl(\tfrac{n^{1-\lambda}}{\Delta_i}\Bigr) \cdot \tilde{O}(\Delta_i^2)
          = \sum_{i:\,\Delta_i \leq n^{\kappa+\lambda}} \tilde{O}(n^{1-\lambda}\Delta_i) = \tilde{O}(n^{1-\lambda} \cdot n^{\kappa+\lambda}) = \tilde{O}(n^{1+\kappa}),
        \end{align*}
        which stays within budget.
  \item Consider the second challenge. Let $\Pcal$ be the partition of $\widetilde{G}$ produced by
        \textsf{MPX}. For each $P \in \Pcal$, we denote by $P_M$ the set of its \emph{non-boundary} supernodes. Take a part $P \in \Pcal$ that is not left open in the resolution of the first
        challenge. Since the edges incident to boundary vertices of $P$ have already been added to $F$, it suffices to detect the outgoing edges of $P_M$ not crossing to the boundary vertices $P_B := P \setminus P_M$ in $P$. By
        definition, every vertex in $P_M$ has all of its $\widetilde{G}$-edges inside $P$, and thus every such outgoing edge is \emph{non-sampled} (absent
        from $\widetilde{G}$). These are the ``hidden'' edges we must find (see Figure~\ref{fig:closepart1}).
        The primary tool for detecting them is the outdetect labeling scheme~\cite{IEWM25}
        (Theorem~\ref{thm:outdetect}). Specifically, if $P_M$ has $f$ outgoing edges reaching the outside of $P$ in the current graph $\widehat{G} - F$,
        then all of them can be recovered by one-shot aggregation of $\tilde{O}(f)$-bit input values along the spanning tree of $P$. This takes $\tilde{O}(n^{\lambda} + f)$ rounds and $\tilde{O}(f |V(G[P])|)$ messages, where $\tilde{O}(n^{\lambda})$ is the
        diameter upper bound of the subgraph of $(V(G), F)$ induced by $P$, and $|V(G[P])|$ denotes the number of \emph{original} vertices of $G$ in that subgraph. Consequently, for the parts $P$ such that $P_M$ emits only $\tilde{O}(n^{\kappa})$ outgoing edges in $\widehat{G} - F$, all of these edges can be detected
        using
        \begin{align*}
          \sum_{P} \tilde{O}(n^{\kappa} |V(G[P])|) = \tilde{O}\Bigl(n^{\kappa}\sum_{P}|V(G[P])|\Bigr) = \tilde{O}(n^{\kappa}\cdot n) = \tilde{O}(n^{1+\kappa})
        \end{align*}
        messages in total, since the parts partition the $n$ original vertices and hence $\sum_{P}|V(G[P])| \leq n$.

        \begin{figure}[t]
          \centering
          \begin{subfigure}[b]{0.48\textwidth}
            \centering
            \includegraphics[width=0.7\textwidth]{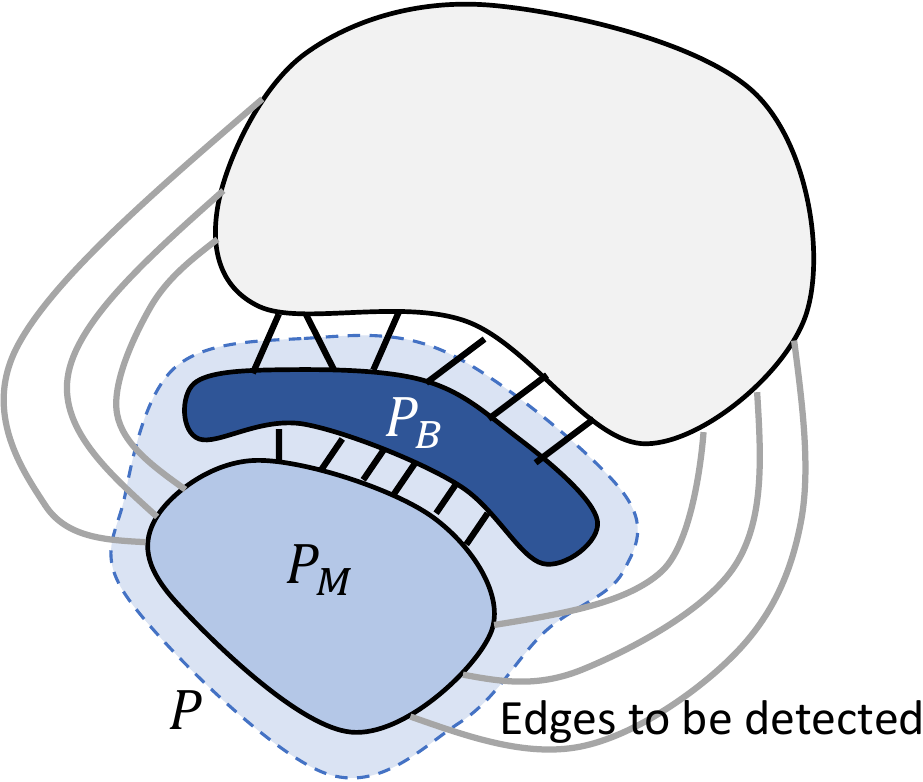}
            \caption{The boundary supernodes $P_B$ and non-boundary supernodes $P_M$ of a part $P$, together with the non-sampled (hidden) outgoing edges of $P_M$.}
            \label{fig:closepart1}
          \end{subfigure}
          \hfill
          \begin{subfigure}[b]{0.48\textwidth}
            \centering
            \includegraphics[width=0.9\textwidth]{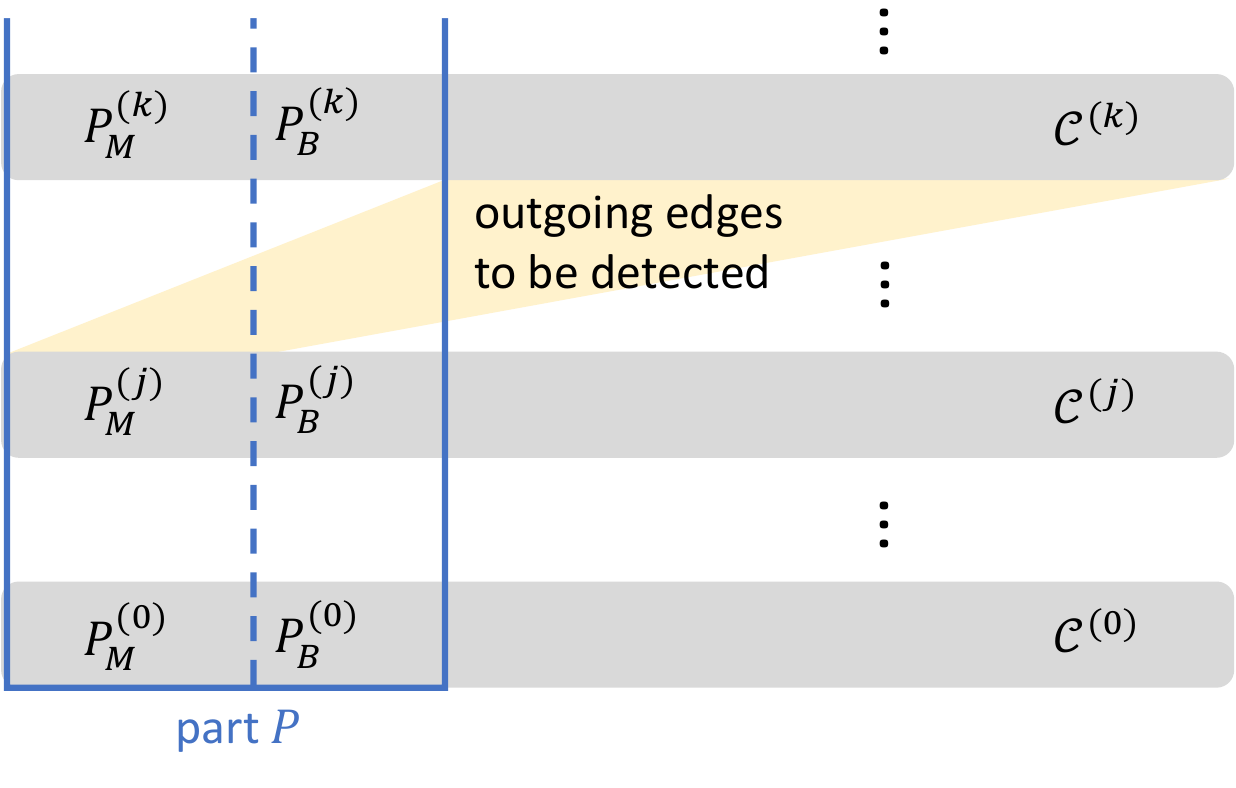}
            \caption{Applying Lemma~\ref{lma:outgoingEdge}: with $S = P_M^{(j)}$ and $U = P^{(k)}$. The}
            \label{fig:closepart2}
          \end{subfigure}
          \caption{Closing an admissible part $P$ in the second challenge.}
          \label{fig:closepart}
        \end{figure}

        It thus remains to bound, for each part we attempt to close, the number of non-sampled outgoing edges from $P_M$ to $V(\widehat{G}) \setminus P_B$ by $\tilde{O}(n^{\kappa})$. Unfortunately, this bound does not hold for all such parts, and thus our algorithm adopts the compromise of closing only \emph{admissible} ones: we call a set of class-$i$ supernodes \emph{$i$-moderate} if it
        has at most $\tilde{O}(n^{2\kappa}/\Delta_i)$ elements.
        Writing $X^{(i)}$ for the set of class-$i$ supernodes in a
        supernode set $X \subseteq V(\widehat{G})$, we say that a part $P$ is admissible if $P^{(i)}$
        is $i$-moderate for every class $i$ (this is condition~(C2) in Section~\ref{sec:partition}).\footnote{In the formal definition of
          Section~\ref{sec:partition}, a part is \emph{admissible} if it satisfies the two conditions~(C1) and~(C2). Here, condition~(C1),
          namely $P_B \cap \Ccal^H = \emptyset$, is precisely the requirement that $P$ is not left open in the resolution of the first
          challenge. Among such parts (which are considered in this context), the admissible ones are exactly
          those additionally satisfying condition~(C2).}

        The key technical lemma for closing all admissible parts is  Lemma~\ref{lma:outgoingEdge}, which intuitively states that a moderate set cannot have many outgoing edges that all escape the sampling. A more precise statement is as follows:
        with high probability, for any pair of a $j$-moderate set $S$ and a $k$-moderate set $U$ admitting more than $n^{\kappa}(\ln n + 1)$ edges crossing between $S$ and $\Ccal^{(k)} \setminus U$, at least one of them is sampled.Fortunately, this lemma is proved by a relatively standard union-bound argument, and is sufficient to bound the
        number of outgoing edges from $P_M$ to $V(\widehat{G}) \setminus P_B$: for each pair of classes $(j, k)$, we take $S = P_M^{(j)}$ (the non-boundary class-$j$
        supernodes of $P$) and $U = P^{(k)}$ (its class-$k$ supernodes). The admissibility of $P$ makes $S$ $j$-moderate and $U$ $k$-moderate, and no edge from
        $S$ to the class-$k$ supernodes outside $P$ is sampled. The lemma therefore bounds these edges by $\tilde{O}(n^{\kappa})$ per class pair, hence by $\tilde{O}(n^{\kappa})$ in total (see Figure~\ref{fig:closepart2}).

        Finally, the parts violating admissibility are few: if a part is non-admissible, for some $i$, it contains $\tilde{\Omega}(n^{2\kappa}/\Delta_i)$ elements in $\Ccal^{i}$. It implies that only $\tilde{O}(n^{1-2\kappa})$ parts can be non-admissible (see the bound on
        $|\Pcal_2|$ in the proof of Lemma~\ref{lma:opendiameter}). We simply leave these open, which adds a
        further $\tilde{O}(n^{1-2\kappa} \cdot n^{\lambda}) = \tilde{O}(n^{1-2\kappa+\lambda})$ to the total open diameter.
\end{itemize}

\subsection{Related Work}

\paragraph{Distributed MST in the \textsf{CONGEST} model.}
There is a long line of research on distributed MST in the \textsf{CONGEST} model.
Starting from the classical GHS algorithm~\cite{GHS83} with $O(n\log n)$ rounds
and $O(m + n\log n)$ messages, subsequent work~\cite{Awerbuch87,GKP98,KP98} progressively
reduced the round complexity to the near-optimal $O(D_G + \sqrt{n}\log^* n)$ bound, matching
the $\tilde\Omega(D_G + \sqrt{n})$ lower bound of~\cite{PR00,SHKKNPPW12} up to polylogarithmic factors.
In the \textsf{CONGEST-KT$_0$} model, simultaneous near-optimality in both rounds
($\tilde{O}(D_G + \sqrt{n})$) and messages ($\tilde{O}(m)$) was achieved by
Pandurangan, Robinson, and Scquizzato~\cite{PRS17}. While this algorithm is randomized, a deterministic
algorithm with the same running time and message complexity was proposed by
Elkin~\cite{Elkin20}. Haeupler, Hershkowitz, and Wajc~\cite{HHW18} also presented a round- and
message-optimal deterministic algorithm following the \emph{low-congestion shortcut} framework of Ghaffari
and Haeupler~\cite{GH16}. That approach allows one to extend the MST result to other problems such as
approximate minimum-cut~\cite{GH16} and single-source shortest path~\cite{ZGYHS22}.
In \textsf{CONGEST-KT$_1$}, the $\tilde{\Omega}(m)$ message lower bound can be circumvented~\cite{KKT15,GK18,GP18}. These results have been already discussed in Section~\ref{subsec:background}.

Mashreghi and King~\cite{MK21} extended the $o(m)$-message paradigm
to the \emph{asynchronous} model, achieving $O(n^{3/2}\log^{3/2} n)$ messages for both MST and broadcast.
Dufoulon, Kutten, Moses~Jr., Pandurangan, and Peleg~\cite{DKMPP22} later gave an almost singularly optimal
asynchronous MST algorithm with $\tilde{O}(m)$ messages
and $\tilde{O}(D_G^{1+\varepsilon} + \sqrt{n})$ rounds.

\paragraph{$o(m)$ messages for other problems and models.}
Sublinear-message computation has been studied beyond MST.
Pai, Pandurangan, Pemmaraju, Riaz, and Robinson~\cite{PPPRR17}
showed that a $2$-ruling set can be computed with $O(n\log^2 n)$ messages,
while computing MIS requires $\Omega(n^2)$.
Pai, Pandurangan, Pemmaraju, and Robinson~\cite{PPPR21} further showed that $(\Delta+1)$-coloring
admits $\tilde{O}(n^{3/2})$-message algorithms using non-comparison-based
randomized techniques,
whereas comparison-based algorithms still require $\Omega(n^2)$ messages.
For graph optimization problems (MVC, MDS, MaxIS),
Dufoulon, Pai, Pandurangan, Pemmaraju, and Robinson~\cite{DPPPR24} proved near-cubic message lower bounds
and separated the complexity of exact from approximate computation.
Very recently, Dufoulon, Pai, Pandurangan, Pemmaraju, and Robinson~\cite{DPPPR25} extended the
message-complexity viewpoint to APSP and related problems.
Censor-Hillel, Haeupler, Kelner, and Maymounkov~\cite{CHKM12} showed that in the \textsf{LOCAL} model,
any algorithm can be transformed to use $\tilde{O}(n)$ messages only with logarithmic overhead in round
complexity. A similar result with no time overhead was also proposed by Bitton, Emek, Izumi, and
Kutten~\cite{BEIK19}.

\subsection{Roadmap}
Section~\ref{sec:prelim} introduces necessary notations and terminologies, and recalls the tools used
in our algorithm.
Section~\ref{sec:MST} introduces the $(\alpha, \beta)$-partition framework and
reduces MST construction to spanner construction.
Section~\ref{sec:partition} presents the efficient distributed construction of an $(\alpha,\beta)$-partition,
which is the main algorithmic contribution of this paper.
Section~\ref{sec:conclude} concludes the paper.

\section{Preliminaries}
\label{sec:prelim}

\subsection{Notations and Terminologies}
Throughout this paper, we denote the vertex set and the edge set of a graph $H$ by $V_{H}$ and
$E_{H}$, respectively. The notation $H' \subseteq H$ means that $H'$ is a subgraph of $H$.
Given a vertex subset $X \subseteq V_H$, let $H[X]$ be the subgraph induced by $X$.
For any edge subset $X \subseteq E(H)$, $H - X$ denotes the graph obtained from $H$ by removing all edges in $X$.
Similarly, for any vertex subset $X \subseteq V(H)$, $H - X$ denotes the graph obtained from $H$ by removing all vertices in $X$ and their incident edges.
For two vertex subsets $X, Y \subseteq V(H)$, let $\partial_{H}(X, Y)$ be the set of edges crossing
between $X$ and $Y$. We also define $\partial_{H}(X)$ as
$\partial_H(X) = \partial_H(X, V(H) \setminus X)$. If either $X = \emptyset$ or $Y = \emptyset$ holds, $\partial_H(X, Y)$ is defined
as $\partial_H(X,Y) = \emptyset$.
Edges in $\partial_H(X)$ are called \emph{outgoing edges} of $X$ in $H$.
We also define $I_{H}(X)$ as the set of edges in $H$ incident to a vertex in $X$.
When $X$ is a singleton $X = \{u\}$, we write $\partial_H(u)$ instead of $\partial_H(\{u\})$, and $I_H(u)$ instead of $I_H(\{u\})$.
For a vertex $u$ in $H$, we denote the set of neighbors of $u$ by
$N_{H}(u) = \{v \in V(H) \mid (u,v) \in E(H) \}$,
and we denote the degree of $u$ by $\Deg_{H}(u)$.
For a vertex subset $X \subseteq V(H)$, we define $\Deg_H(X) = \sum_{u \in X} \Deg_H(u)$.
For any graph $H$ and two vertices $u, v \in V(H)$, $\Dist_H(u, v)$ denotes the distance between $u$ and $v$ in $H$.
Let $D_H$ denote the diameter of $H$, i.e.,
$D_H := \max_{u, v \in V(H)} \Dist_H(u,v)$.
A \emph{partition} of $V(H)$ is a disjoint collection of subsets of $V(H)$ such that their union is $V(H)$.
For a partition $\Pcal$, each $P \in \Pcal$ is referred to as a \emph{part} of $\Pcal$.

\subsection{\textsf{CONGEST-KT$_1$} Model}
As the computational model, we assume the \textsf{CONGEST-KT$_1$} model.
Let $G$ be the input graph (communication topology) with $n$ vertices and $m$ edges.\footnote{Since we consider the MST problem, the input graph is formally equipped with an edge-weight function $\mathsf{wt} : E(G) \to \mathbb{N}$. However, the weights are relevant only to the reduction from spanners to MST (Lemma~\ref{lma:mst-from-spanner}), which we take as a black box from prior work. The rest of our construction operates purely on the unweighted structure of $G$, and we therefore do not make the weight function explicit.}
Computation on $G$ proceeds in synchronous rounds, where in each round every node can perform arbitrary local computation and sends a (possibly different) message of $O(\log n)$ bits to each of its neighbors.
Each vertex has a unique identifier (ID), where IDs are drawn from the set of
nonnegative integers. We assume that an ID can be represented using $O(\log n)$ bits.
When we explicitly refer to the ID of a vertex $u$, we write $\mathsf{id}(u)$.
In the \textsf{KT$_1$} setting, each vertex initially knows the IDs of all its neighbors.

As implicitly assumed in prior work, we also assume that all the vertices initially know the precise value of $n$. Note that this assumption can be relaxed to the knowledge of any constant approximation of $n$ without degrading the asymptotic performance
of the algorithm.

\subsection{Outdetect Labeling Scheme}
The \emph{outdetect labeling scheme}~\cite{IEWM25} is a labeling scheme for detecting outgoing edges of a given
subgraph $H$ of the input graph $G$ with a technique similar to
the \emph{graph sketch} by Ahn, Guha, and McGregor~\cite{AGMR12,AGM12} and \emph{TestOut} by King, Kutten, and Thorup~\cite{KKT15}.
The outdetect labeling scheme assigns each vertex $u \in V(H)$ a short label $L(u)$ of
$\tilde{O}(f)$ bits, where $f$ is a design parameter.
Given the bitwise XOR
\begin{align*}
  L(V(H)) = \oplus_{v \in V(H)} L(v),
\end{align*}
the scheme admits recovering all edges in $\partial_G(V(H))$ (in the form of endpoint pairs)
if $|\partial_G(V(H))| \leq f$, or otherwise outputting a subset of $f$ edges in $\partial_G(V(H))$
together with the detection that $|\partial_G(V(H))| > f$.
Label construction is fully decentralized in the \textsf{KT$_1$} setting: each vertex $u$ can
compute its label $L(u)$ locally. Suppose that $H$ has a spanning tree $T_H$ of height $h$.
Then, using standard pipelined aggregation and broadcast along $T_H$, every vertex in $V(H)$ can compute $L(V(H))$
within $\tilde{O}(h + f)$ rounds and detect the outgoing edges in $\partial_G(V(H))$.
This procedure uses $\tilde{O}(|V(H)|f)$ messages.

We formally state this subroutine as follows.

\begin{theorem} \label{thm:outdetect}
  Let $G$ be any graph, and let $H \subseteq G$ be a subgraph.
  Suppose that a spanning tree $T_H$ of height $h$ is given.
  For any $f > 0$, there exists a randomized CONGEST algorithm $\textsf{OutDetect}(f)$ such that
  every vertex in $V(H)$ detects all edges in $\partial_G(V(H))$ if $|\partial_G(V(H))| \leq f$, or
  otherwise outputs a subset of $f$ edges in $\partial_G(V(H))$ together with the detection that $|\partial_G(V(H))| > f$.
  The algorithm runs in $\tilde{O}(h + f)$ rounds, uses $\tilde{O}(|V(H)|f)$ messages, and succeeds
  with probability at least $1 - O(1/n^3)$.
\end{theorem}

\subsection{Low-Diameter Decomposition}

\begin{definition}[Low-Diameter Decomposition]
  Let $H=(V(H),E(H))$ be an undirected graph,
  and let $d > 0$. A partition
  $\Scal=\{S_1,\dots,S_k\}$ of $V(H)$
  is called a $d$-low diameter decomposition ($d$-LDD)
  if the following conditions hold:
  \begin{enumerate}
    \item For every cluster $S_i\in\Scal$,
          $\max_{u,v\in S_i}\Dist_{H[S_i]}(u,v) = O(d \log n)$.
    \item At most $O(n / d)$ vertices have an incident edge crossing different parts.
  \end{enumerate}
\end{definition}

We refer to a vertex $v \in V(H)$ having an incident edge crossing different parts as a \emph{boundary vertex}.
It is well known that there exists a randomized CONGEST algorithm for computing a $d$-LDD of
the input graph.

\begin{theorem}[Miller, Peng, Xu~\cite{MPX13}] \label{thm:diameter-d-decomposition}
  Given an input graph $H = (V(H),E(H))$ and $d > 0$, there exists a randomized algorithm
  $\textsf{MPX}(d)$
  for computing a partition $\Qcal$ of $V(H)$ satisfying the following conditions:
  \begin{enumerate}
    \item $\max_{Q \in \Qcal} D_{H[Q]} \leq O(d \log n)$ holds with probability $1 - O(1/n)$.
    \item Each vertex $v$ becomes a boundary vertex with probability at most $2d^{-1}$.
  \end{enumerate}
  The running time of the algorithm is $O(d \log n)$ rounds with probability $1 - O(1/n)$,
  and the message complexity is $O(|E(H)|)$.
\end{theorem}

\section{MST Construction via \texorpdfstring{$(\alpha, \beta)$}{(alpha, beta)}-Partition}
\label{sec:MST}

\subsection{\texorpdfstring{$(\alpha, \beta)$}{(alpha, beta)}-Partition}

For $\alpha \geq 1$ and $\beta \geq 0$, a spanning subgraph $H' \subseteq H$ is called an \emph{$(\alpha,\beta)$-spanner} of $H$ if for all $u, v \in V(H)$,
\begin{align*}
  \Dist_{H'}(u,v) \le \alpha\,\Dist_H(u,v) + \beta.
\end{align*}

Let $H$ be a graph, and let $\Pcal$ be a partition of $V(H)$.
For a set of edges $F \subseteq E(H)$, define $H_F = (V(H), F)$.
We say that a part $P \in \Pcal$ is \emph{closed with respect to $F$}
if $\partial_H(P) \subseteq F$ holds, and otherwise we say that $P$ is \emph{open with respect to $F$}.
Let $\Pcal^{\mathrm{open}}_F$ (resp., $\Pcal^{\mathrm{closed}}_F$) denote the family of all parts that are open (resp., closed) with respect to $F$.
We now define the notion of an $(\alpha, \beta)$-partition.

\begin{definition} \label{dfn:abc-partition}
  Let $H=(V(H), E(H))$ be a graph, and let $F \subseteq E(H)$ be an edge subset.
  An \emph{$(\alpha, \beta)$-partition} of $H$ is a pair $(\Pcal, F)$ of a partition $\Pcal$ and an edge subset $F \subseteq E(H)$
  satisfying the following two conditions:
  \begin{enumerate}
    \item \label{enum:part-diameter} For every $P \in \Pcal$, $D_{H_F[P]} \leq \alpha$.
    \item \label{enum:open-part} The open parts with respect to $F$ satisfy
          $\sum_{P \in \Pcal^{\mathrm{open}}_F} (D_{H_F[P]} + 1) \leq \beta$.
  \end{enumerate}
\end{definition}

\begin{lemma} \label{lma:spanner-from-partition}
  Let $H=(V(H), E(H))$ be a graph, and let $(\Pcal, F)$ be an $(\alpha, \beta)$-partition
  of $H$. Suppose that each vertex $v$ knows the set $\partial_H(v) \cap F$ and also knows the part $P \in \Pcal$ such that $v \in P$.
  Then there exists a randomized CONGEST algorithm for constructing an $((\alpha + 1), \beta)$-spanner of $H$ that runs in $\tilde{O}(\alpha + \beta)$ rounds and uses $\tilde{O}(n)$ messages.
\end{lemma}

\begin{proof}
  We first describe the construction algorithm. Initially, each part $P \in \Pcal$ determines whether it is open
  or closed with respect to $F$. By definition, $P$ is open if and only if
  $\partial_{H - F}(P) \neq \emptyset$. This condition can be verified by running
  \textsf{OutDetect} with parameter $f = 1$ in $H - F$ (recall that each vertex initially knows which of its incident edges belong to $F$).

  The algorithm then iteratively merges open parts by detecting their outgoing edges.
  In each iteration, it applies \textsf{OutDetect} with $f = 1$ in $H - F$ to every part that is still open.
  Whenever two open parts are detected to be adjacent, the algorithm merges them by adding the detected
  crossing edge to $F$. Otherwise, the part becomes closed.
  Since the number of open parts decreases by at least a constant fraction in each iteration,
  $O(\log n)$ iterations suffice to close all parts. In the following argument, let $F'$ be the
  edge set $F$ after this merging step.

  The message complexity of each iteration is $\tilde{O}(n)$ because $f = 1$.
  Regarding the running time, the first invocation of \textsf{OutDetect} takes $\tilde{O}(\alpha)$ rounds by
  condition~(\ref{enum:part-diameter}) of Definition~\ref{dfn:abc-partition}.
  In the subsequent iterations, it is invoked only for open parts.
  By condition~(\ref{enum:open-part}), the total diameter of open parts at any iteration is bounded by $\beta$.
  Hence, each iteration takes $\tilde{O}(\beta)$ rounds, and the total running time is $\tilde{O}(\alpha + \beta)$.

  We now show that the resulting graph $H' = (V(H), F')$ is an $((\alpha+1), \beta)$-spanner of $H$.
  Let $\Pcal'$ denote the partition resulting from the merging process.
  Recall that $\Pcal'$ consists only of closed parts (with respect to $F'$).
  Let $u, v \in V(H)$ be arbitrary vertices, and consider a shortest $u$--$v$ path $\pi$ in $H$. We decompose $\pi$ into maximal subpaths contained in a single part of $\Pcal'$:
  \begin{align*}
    \pi = B_1 \circ e_1 \circ B_2 \circ e_2 \circ \cdots \circ e_{t-1} \circ B_t,
  \end{align*}
  where each $B_j$ lies entirely in a single part $P(B_j) \in \Pcal'$,
  and each $e_j$ is an edge between $P(B_j)$ and $P(B_{j+1})$.
  Let $a_j$ and $b_j$ denote the first and last vertices of $B_j$, respectively, and set $a_{t+1} = v$.
  Since every part in $\Pcal'$ is closed with respect to $F'$, every crossing edge $e_j$ belongs to $F' \subseteq E(H')$.

  We bound $\Dist_{H'}(u,v)$ by exhibiting a $u$--$v$ path $\pi'$ in $H'$ obtained from $\pi$: within each part $P(B_j)$ we replace $B_j$ by a shortest $a_j$--$b_j$ path of $H'[P(B_j)]$, and between consecutive parts we keep the crossing edges $e_j \in F'$.
  Recalling $a_1 = u$ and $a_{t+1} = v$, the length of $\pi'$ satisfies $\Dist_{H'}(u,v) \leq \sum_{j=1}^{t} \Dist_{H'}(a_j, a_{j+1})$, and since $a_1, \dots, a_{t+1}$ are consecutive vertices of $\pi$, we have $\sum_{j=1}^{t} \Dist_H(a_j, a_{j+1}) = \Dist_H(u,v)$.
  We split the sum according to the origin of $P(B_j)$, separating the segments lying in an originally closed part from the rest:
  \begin{align*}
    \Dist_{H'}(u,v) \leq \sum_{j:\, P(B_j) \in \Pcal_F^{\mathrm{closed}}} \Dist_{H'}(a_j, a_{j+1})
    \;+\; \sum_{j:\, P(B_j) \notin \Pcal_F^{\mathrm{closed}}} \Dist_{H'}(a_j, a_{j+1}),
  \end{align*}
  and bound the two sums separately. The first sum is bounded as
  \begin{align*}
    \sum_{j:\, P(B_j) \in \Pcal_F^{\mathrm{closed}}} \Dist_{H'}(a_j, a_{j+1})
    \leq (\alpha+1) \sum_{j:\, P(B_j) \in \Pcal_F^{\mathrm{closed}}} \Dist_H(a_j, a_{j+1})
    \leq (\alpha+1)\,\Dist_H(u,v).
  \end{align*}
  For the second sum, we may assume that each contributing part is entered exactly once, because whenever $\pi$ enters a part twice we can shortcut within that part and thereby further shorten $\pi'$. As these parts originate from $\Pcal_F^{\mathrm{open}}$ and each open bundle adds only one crossing edge, condition~(\ref{enum:open-part}) bounds the second sum by $\sum_{P \in \Pcal_F^{\mathrm{open}}} (D_{H_F[P]} + 1) \leq \beta$.
  Combining the two bounds yields $\Dist_{H'}(u,v) \leq (\alpha+1)\,\Dist_H(u,v) + \beta$. Hence $H'$ is an $((\alpha+1), \beta)$-spanner of $H$.
\end{proof}

\subsection{From Spanner to MST}

\begin{restatable}[Ghaffari and Kuhn, and Gmyr and Pandurangan]{lemma}{LmaMSTfromSpanner}
  \label{lma:mst-from-spanner}
  Suppose that an $(\alpha,\beta)$-spanner $G^{\ast}$ of $G$
  with $|E(G^{\ast})|=M$ is constructed. Then one can compute an MST of $G$ in
  \[
    \tilde{O}\bigl(\alpha D_G + \beta + \sqrt{n}\bigr)
  \]
  rounds using $\tilde{O}(M)$ messages in the \textsf{CONGEST-KT$_1$} model.
\end{restatable}

\sloppy{
  The outline of the proof is given in the appendix.
  Combining Lemma~\ref{lma:mst-from-spanner} with Lemma~\ref{lma:spanner-from-partition}, our
  goal reduces to constructing an edge subset $F$ of size $\tilde{O}(n^{1 - 3\kappa/2})$ together
  with a $(\tilde{O}(n^{\kappa/2}), \tilde{O}(n^{1-3\kappa/2}))$-partition with respect to $F$.
}

\section{Message-Efficient Construction of \texorpdfstring{$(\alpha, \beta)$}{(...)}-Partition}
\label{sec:partition}

\subsection{Degree-Adaptive Clustering}

In this section, we present a clustering algorithm \textsf{DCluster}, which outputs an edge
subset $F \subseteq E(G)$ and a partition $\Ccal$ of $V(G)$ naturally induced by the
connected components of $(V(G), F)$. Let $i_{\max} = \lceil \log_2 n \rceil - 1$.
For $0 \leq i \leq i_{\max}$, we define $\Delta_i = 2^i$, and
$X_i = \{v \in V(G) \mid \Deg_G(v) > \Delta_i\}$. Precisely, the output
partition $\Ccal$ consists of $i_{\max} + 1$ collections of subpartitions $\Ccal^{(i)} \subseteq \Ccal$
($i \in [0, i_{\max}]$) satisfying the following properties:
\begin{enumerate}
  \item The collections $\Ccal^{(i)}$ for all $i$ are mutually disjoint.
  \item For any vertex $v$ in a cluster $C \in \Ccal^{(i)}$,
        $\Deg_G(v) \leq 2\Delta_{i}$.
  \item The diameter of the subgraph of $(V(G),F)$ induced by each cluster is at most $2\lceil \log_2 n \rceil - 2$.
  \item $|\Ccal^{(i)}| \leq 12 n \ln n / \Delta_i$.
\end{enumerate}

Algorithm~\ref{alg:clustering} shows the pseudocode of \textsf{DCluster}.
It processes vertices in $X_i$ in decreasing order of $i$.
In iteration $i$, it makes vertices in $X_i$ join some cluster.
All vertices that join a cluster are removed from $G_i$.

Let $G_i$ be the remaining graph at the beginning of iteration $i$.
In iteration $i$, each vertex in $V(G_i)$ is independently sampled
as a cluster center with probability $\min\{1, 6\ln n/\Delta_i\}$. The
sampled set is denoted by $W_i$ in Algorithm~\ref{alg:clustering}.
If a vertex $v \in V(G_i) \cap X_i$ has a neighboring center $w \in W_i$,
then it joins the cluster of $w$ by adding $\{v, w\}$ to $F$
(if two or more centers are neighboring, an arbitrary one is chosen).
Otherwise, $v$ samples $O(\log n)$ vertices from $N_G(v)$ independently.
(We emphasize that they are sampled from $N_G(v)$, but not from $N_{G_i}(v)$.)
If among them there exists a vertex $v'$ that has already been removed (i.e., already
joined a cluster in an earlier iteration), then $v$ connects to that cluster by adding
$\{v, v'\}$ to $F$. The vertices that are chosen as centers or join some cluster
are removed from $G_i$. Finally, the clusters centered at vertices in $W_i$ are output
as $\Ccal^{(i)}$.

One can show that all vertices in $X_i$ are removed during iteration $i$
with high probability. Intuitively, for any $v \in X_i$, there are two
possible cases: (1) a majority of its neighbors still remain. Since $v \in X_i$, its degree
in $G$ is at least $\Delta_i$, and thus its degree in the current graph
is at least $\Delta_i/2$. Hence, at least one neighbor is sampled with high probability. (2)
A majority of its neighbors have already been removed. Then the sampling by $v$ picks up at least one
neighbor that has already been removed. In either case, $v$ joins a cluster.

\begin{algorithm}[t]
  \caption{$\textsf{DCluster}(G)$ : Degree-Adaptive Clustering}
  \label{alg:clustering}
  \begin{algorithmic}[1]
    \State $\mathsf{sample}(X, p)$: sample each element in $X$ independently with probability $p$
    \State $F \gets \emptyset$;

    \State $G_{i_{\max}} \gets G$
    \For{$i \gets i_{\max}$ \textbf{downto} $0$}
    \State $W_{i} \gets \mathsf{sample}(V(G_i), \min\{1, 6\ln n/\Delta_i\})$

    \ForAll{$v\in X_i\setminus W_i$}
    \If{$N_G(v)\cap W_i\neq\emptyset$}
    \State choose any $w\in N_G(v)\cap W_i$
    \State $F\gets F\cup\{\{v,w\}\}$
    \Else
    \State $Y \gets \mathsf{sample}(N_G(v), \min\{1, 6\ln n/\Delta_i\})$
    \If{$\exists y\in Y \setminus V(G_i)$}
    \State choose any such $y$
    \State $F\gets F\cup\{\{v,y\}\}$
    \EndIf
    \EndIf
    \EndFor
    \If{$i>0$}
    \State $G_{i-1}\gets G_i\setminus (X_i\cup W_i)$
    \EndIf
    \EndFor
    \For{$i \gets i_{\max}$ \textbf{downto} $0$}
    \State $\Ccal^{(i)} \gets \text{all clusters centered at vertices in $W_i$}$
    \EndFor
  \end{algorithmic}
\end{algorithm}

The correctness of \textsf{DCluster} is presented below:
\begin{lemma}
  \label{lma:correctnessDCluster}
  Algorithm \textsf{DCluster} outputs the set $\{\Ccal^{(i)}\}_{0 \leq i \leq i_{\max}}$ and $F \subseteq E(G)$ satisfying properties (1)--(4) with probability $1 - O(\log n / n^2)$.
\end{lemma}

\begin{proof}
  Property (1) is immediate. Since $|V(G_i)| \leq n$ trivially holds, we have
  $\mathbb{E}[|W_i|] \leq 6 n \ln n / \Delta_i =: M$.
  Since $\mathbb{E}[|W_i|] \leq M$, the Chernoff bound in the form
  $\Pr[X \geq 2M] \leq \exp(-M/3)$ (which holds for any $M \geq \mathbb{E}[X]$) yields
  \begin{align*}
    \Pr\left[|W_i| > \frac{12 n \ln n}{\Delta_i}\right]
    \leq \exp\left(-\frac{M}{3}\right)
    = \exp\left(-\frac{2 n \ln n}{\Delta_i}\right)
    \leq n^{-2},
  \end{align*}
  where the last inequality uses $\Delta_i \leq n$.
  By a union bound over all $i_{\max} + 1 \leq \lceil \log_2 n \rceil$ iterations,
  property (4) holds simultaneously for all $i$ with probability $1 - O(\log n / n^2)$.

  We prove property (2) by showing that vertices in $X_i$ are removed during iteration $i$.
  The proof is by induction on $i$. Let $p_i = \min\{1, 6\ln n / \Delta_i\}$ for short.
  (Basis) If $i = i_{\max}$, then $X_{i+1}$ is empty, and the statement holds trivially.
  (Inductive step) Suppose as the induction hypothesis that vertices in $X_{i+1}$ have been removed at the beginning of iteration $i$,
  and consider a vertex $v \in X_i$. If $|N_{G_i}(v)| \geq |N_{G}(v)|/2$, the probability
  that $v$ cannot find any neighbor in $W_i$
  is bounded by $(1 - p_i)^{|N_G(v)|/2} \leq (1 - p_i)^{\Delta_i/2} \leq n^{-3}$. Otherwise, $v$
  has at least $\Delta_i/2$ removed neighbors, and thus at Line~11 of the algorithm it samples
  a vertex $y$ that has already been removed with probability at least $1 - n^{-3}$. That is, $v$ joins $y$'s cluster.
  In either case, $v$ is removed.

  Finally, we prove property (3).
  It is easy to show that each cluster contains exactly
  one vertex in $W = \bigcup_{0 \leq i \leq i_{\max}} W_i$. Since each vertex $v \in X_i$ is
  connected to a vertex in $W_i$ or a vertex already removed in an earlier iteration (which belongs to $X_j$ for some $j > i$),
  it has distance at most $i_{\max} - i$ from its cluster center.
  Since $i \geq 0$, this distance is at most $i_{\max} = \lceil \log_2 n \rceil - 1$. This implies that the diameter of each cluster is at most $2(\lceil \log_2 n \rceil - 1)$, and thus
  property (3) also holds.
\end{proof}

While Algorithm~\ref{alg:clustering} is stated in a centralized manner, it is easy to implement
it in the CONGEST model. The steps in which communication is required are at Lines~7--9 and 11--14.
For Lines~7--9, sampled vertices send the information ``I'm a center'' to all neighbors,
and each vertex $v$ joins a cluster by replying to it. Lines~11--14 are similar: each vertex
probes randomly selected neighbors by sending messages, and the vertices receiving those messages
reply whether they have joined a cluster. We also bound the message complexity
of \textsf{DCluster} as follows:

\begin{lemma}
  \label{lma:message-complexity-DCluster}
  The message complexity of \textsf{DCluster} is at most $48 n \ln n \lceil \log_2 n \rceil$
  with probability $1 - O(\log n / n^2)$.
\end{lemma}

\begin{proof}
  We analyze the two communication steps in each iteration $i$ separately. (\textbf{Lines~7--9})
  Each vertex in $W_i$ broadcasts the message ``I'm a center'' to all its neighbors in $G_i$.
  Since $G_i$ contains no vertex in $X_{i+1}$, every vertex of $G_i$ has degree at most $2\Delta_i$. The number of messages in this step is at most $|W_i| \cdot 2\Delta_i$.
  By Lemma~\ref{lma:correctnessDCluster}, $|W_i| \leq 12 n \ln n / \Delta_i$ holds for all $i$
  simultaneously with probability $1 - O(\log n / n^2)$,
  which implies at most $24 n \ln n$ messages per iteration. (\textbf{Lines~11--14})
  Each vertex $v \in X_i \cap V(G_i)$ sends probe messages to neighbors sampled from $N_G(v)$
  with probability $6 \ln n / \Delta_i$ per neighbor, and receives one reply per probe.
  Since $\Deg_G(v) \leq 2\Delta_i$, we have $\mathbb{E}[|Y_v|] \leq 12 \ln n$. Thus the expected number of messages per vertex is at most $24 \ln n$.
  Summing over at most $n$ vertices gives expected total at most $24 n \ln n$ per iteration.
  Applying a Chernoff bound to the sum of independent sampling variables,
  the actual total exceeds $48 n \ln n$ with probability at most $e^{-8n \ln n} = n^{-8n}$ per iteration. Taking a union bound over $i_{\max} + 1 \leq \lceil \log_2 n \rceil$,
  the total message complexity is at most $48 n \ln n \cdot \lceil \log_2 n \rceil$ with probability
  at least $1 - O(\log n /n^2)$.
\end{proof}

\subsection{Balancing Size} \label{subsec:balancing}

Since the clusters in $\Ccal^{(i)}$ are not necessarily balanced in size, some of them might be much
larger than $\Delta_i$. As a post-processing step, we split a large cluster $C$ into small clusters of size $\Theta(\Delta_i)$.
To preserve the connectivity of each cluster, we allow the resulting clusters to overlap slightly. To this end, we use
the algorithm \textsc{Split} from~\cite{IKNS22}, which splits a rooted spanning tree $T_C$ of $C$ in a bottom-up fashion.
It outputs a collection of subtrees that have asymptotically the same size and may share their roots.

\begin{lemma}[Izumi, Kitamura, Naruse, Schwartzman~\cite{IKNS22}] \label{lma:treesplit}
  Let $T$ be a spanning tree of height $h$, and let $k > 1$ be an integer parameter.
  There exists a deterministic CONGEST algorithm that outputs a collection of subtrees
  $\{T_1, T_2, \dots, T_j\}$ such that the following conditions hold:
  \begin{enumerate}
    \item For every $i \in [1, j]$, $k \leq |V(T_i)| \leq 3k$.
    \item For any two distinct indices $i_1, i_2 \in [1, j]$, $V(T_{i_1}) \cap V(T_{i_2})$
          is either empty or a singleton consisting of the common root of $T_{i_1}$ and $T_{i_2}$.
  \end{enumerate}
  Let $h' = \max_{1 \leq i \leq j} D_{T_i}$. The running time of this algorithm is $O(h')$.
  The message complexity is $O(|V(T)|)$.
\end{lemma}

In our case, the subgraph induced by each cluster $C \in \Ccal^{(i)}$ is equipped with an $O(\log n)$-diameter spanning tree, so that we can apply this algorithm with $k = \Delta_i$. To maintain vertex-disjointness of clusters, we replace each vertex shared by $x$ subtrees with a virtual
star graph with $x$ leaves simulated by their common root. Let $G'$ be the graph after
this replacement. Since all the large clusters in $\Ccal^{(i)}$ are collectively split into at most
$n / \Delta_i$ subclusters, the increase in the number of clusters is bounded by $n / \Delta_i$.
Combined with Property 4 of \textsf{DCluster}, which bounds the original cluster count by
$12n \ln n / \Delta_i$, the total number of clusters in $\Ccal^{(i)}$ is at most
$12n \ln n / \Delta_i + n / \Delta_i \leq 13n \ln n / \Delta_i$.
The increase in the number of vertices and edges is also bounded by the number of new clusters
created, i.e., by $O(n)$. Hence, the asymptotic size of $G'$ (in terms of the number of
vertices and edges) is the same as $G$.

Note that $G$ can simulate any algorithm running on $G'$
with the same complexity. In addition, once we construct
an $(\alpha, \beta)$-spanner on $G'$, the corresponding subgraph in $G$ is also an
$(\alpha, \beta)$-spanner. Hence, in the following argument we focus on the construction of an
$(\alpha, \beta)$-partition in $G'$. To avoid introducing further notations,
we abuse notation and use $G$ to denote the graph after this replacement, where $V(G)$ is partitioned
into a family of disjoint cluster sets $\{\Ccal^{(i)}\}_{0 \leq i \leq i_{\max}}$.
We summarize the properties of the clustering in the following lemma.

\begin{lemma} \label{lma:balancedcluster}
  Let $G$ be the graph after the replacement described in this subsection.
  Then there exists a CONGEST algorithm that outputs a collection of disjoint cluster sets
  $\{\Ccal^{(i)}\}_{0 \leq i \leq i_{\max}}$ and an edge subset $F \subseteq E(G)$ satisfying the
  following properties with probability at least $1 - O(\log n / n^2)$:
  \begin{enumerate}
    \item $\Ccal^{(i)}$ for all $i$ are mutually disjoint.
    \item For any vertex $v$ in a cluster $C \in \Ccal^{(i)}$,
          $\Deg_G(v) \leq 2\Delta_{i}$.
    \item The diameter of the subgraph of $(V(G),F)$ induced by each cluster is at most $2\lceil\log_2 n\rceil - 2$.
    \item $|\Ccal^{(i)}| \leq 13n \ln n / \Delta_i$.
    \item For every $C \in \Ccal^{(i)}$, $|C| \leq 3\Delta_i$.
  \end{enumerate}
  The algorithm runs in $O(\log n)$ rounds, and uses $\tilde{O}(n)$ messages.
\end{lemma}

As explained in Section~\ref{subsec:outline}, the following partitioning algorithm works on
top of $\widehat{G}$, which is the multigraph with self-loops obtained from $G$ (of this lemma) by contracting each cluster $C \in \Ccal := \bigcup_i \Ccal^{(i)}$ into a supernode.
In the following subsection, we identify $V(\widehat{G})$ with $\Ccal$, and use the terms ``cluster'' and ``supernode'' interchangeably. When focusing on the vertex set of the original
graph $G$, we often abuse notation and identify a vertex subset $X \subseteq \Ccal$ of $\widehat{G}$ with its union.

\subsection{Partitioning Algorithm}

We present our main partitioning algorithm $\textsf{Partition}(\kappa, \lambda)$
($0 \leq \lambda \leq \kappa \leq 1/2$), which provides an $(\alpha ,\beta)$-partition with
$\alpha = \tilde{O}(n^{\lambda})$ and $\beta = \tilde{O}(n^{1 - \kappa - \lambda} + n^{1 - 2\kappa + \lambda})$. The algorithm attains $\tilde{O}(\alpha + \beta + n^{\kappa})$ round complexity
and $\tilde{O}(n^{1+\kappa})$ message complexity. The algorithm focuses on the construction of an $(\alpha, \beta)$-partition of $\widehat{G}$ because it trivially provides a
$(\tilde{O}(\alpha), \tilde{O}(\beta))$-partition of the original graph $G$.

Here, we formally define the notation mentioned in Section~\ref{subsec:outline}.
Given a set of clusters $X \subseteq \Ccal$, we write $X^{(i)}$ for the subset consisting of all class-$i$ supernodes,
i.e., $X^{(i)} := X \cap \Ccal^{(i)}$. Given a partition $\Pcal$ of $\widehat{G}$ with boundary vertices $B \subseteq V(\widehat{G})$, $P_B$ denotes all the boundary vertices in $P \in \Pcal$ (i.e. $P_B = P \cap B$), and define $P_M = P \setminus P_B$.
Let $t = \lfloor \log(n^{\kappa+\lambda}\log n) \rfloor$, so that $\Delta_i > n^{\kappa+\lambda}\log n$ if and
only if $i > t$. We define $\Ccal^H = \bigcup_{i > t} \Ccal^{(i)}$, the set of
\emph{high-degree} clusters of degree class $\Delta_i > n^{\kappa+\lambda}\log n$.
A set $X$ of supernodes is said to have \emph{moderate size for class $i$}, or to be \emph{$i$-moderate},
if $X \subseteq \Ccal^{(i)}$ and $|X| \leq n^{2\kappa}\log n / \Delta_i$. A part $P \in \Pcal$ is called
\emph{admissible} if it satisfies both of the following conditions:
\begin{itemize}
  \item[(C1)] $P_B \cap \Ccal^H = \emptyset$. That is, $P$ has no high-degree cluster on its boundary.
  \item[(C2)] $P^{(i)}$ is $i$-moderate for every $0 \leq i \leq i_{\max}$ (equivalently, $|P^{(i)}| \leq n^{2\kappa}\log n / \Delta_i$ for every $i$).
\end{itemize}
Note that condition~(C2) implies $|P^{(i)}| = 0$ whenever $\Delta_i > n^{2\kappa}\log n$, i.e., an admissible part
contains no cluster of degree class exceeding $n^{2\kappa}\log n$.

Our algorithm works as follows:
\begin{enumerate}
  \item Run the clustering algorithm of Lemma~\ref{lma:balancedcluster}.
  \item Each vertex in a cluster $C \in \Ccal^{(i)}$ samples its incident edges independently with
        probability $q_i = \min\Bigl\{1, \frac{3 n^{\kappa} \log n}{\Delta_i}\Bigr\}$. Let $J$ be the set of
        all sampled edges, and let $\widetilde{G} = (V(\widehat{G}), J)$. All the edges in $J$ are added to $F$.
  \item Apply $\textsf{MPX}(n^{\lambda})$ to $\widetilde{G}$.
  \item Let $\Pcal$ be the partition of $V(\widetilde{G})$ ($ = V(\widehat{G})$) output by
        the third step. For each admissible $P\in \Pcal$, apply
        one of the following processes:
        \begin{itemize}
          \item If $\Deg_{\widehat{G}}(P) \leq |V(G[P])| n^{\kappa} (\log n + 1)^3$ holds,
                add all the edges in $I_{\widehat{G}}(P)$ to $F$ (where $|V(G[P])|$ represents
                the number of the original vertices (i.e. vertices in $G$) contained in $P$).
          \item Otherwise, add all the edges in $I_{\widehat{G}}(P_B)$ to $F$. Then apply $\textsf{OutDetect}(n^{\kappa} (\log n + 1)^3)$ to $P_M$ on $\widehat{G} - F$. If the algorithm detects all the edges in
                $\partial_{\widehat{G} - F}(P)$, they are added to $F$.
        \end{itemize}
\end{enumerate}

It is not difficult to see that this algorithm is easily implemented in the \textsf{CONGEST-KT$_1$} model.
The first step is obvious from Lemma~\ref{lma:balancedcluster}. The second step is implemented
by independent random
sampling of edges incident to each vertex $u$ in a class-$i$ cluster with probability $q_i$. The third step
mostly follows Theorem~\ref{thm:diameter-d-decomposition}, but we need to be careful because we run it on the contracted graph $\widetilde{G}$. Fortunately, this does not cause any issue: The algorithm $\textsf{MPX}(d)$ is implemented by running the multi-source BFS once. Hence it is easily executed on top of any contracted graph. For the fourth step, each part $P$ checks whether
the condition $\Deg_{\widehat{G}}(P) \leq |V(G[P])| n^{\kappa} (\log n + 1)^3$ holds.
Since $P$ admits a spanning tree of height $O(n^{\lambda} \log n)$, the check process is
implemented using the standard aggregation. If the condition holds, every vertex in $P$ adds
all of its incident edges to $F$ (note that the information of $\{u, v\} \in F$ is shared with
$u$ and $v$ by their direct communication). Otherwise, $P$ activates $\textsf{OutDetect}$,
which follows Theorem~\ref{thm:outdetect}.

\subsection{Correctness of Algorithm \textsf{Partition}}
We show that the output $\Pcal$ and $F$ satisfy the conditions of
Lemma~\ref{lma:spanner-from-partition}.

\begin{lemma} \label{lma:highsize}
  With probability at least $1 - O(1/\log n)$,
  $|B^{(i)}| \leq 52(\log n + 1)^3 \cdot n^{1 - \lambda} / \Delta_i$ holds for any $i \in [0, i_{\max}]$.
\end{lemma}

\begin{proof}
  By Theorem~\ref{thm:diameter-d-decomposition}, we have
  $\mathbb{E}[|B^{(i)}|] \leq 2 n^{-\lambda} |\Ccal^{(i)}|$.
  By Lemma~\ref{lma:balancedcluster}(4), $|\Ccal^{(i)}| \leq 13n \ln n / \Delta_i$,
  and thus $\mathbb{E}[|B^{(i)}|] \leq 26 n^{1 - \lambda} \ln n / \Delta_i$.
  Applying Markov's inequality, we obtain
  \begin{align*}
    \Pr\left[|B^{(i)}| > 52(\log n + 1)^2 \ln n \cdot \frac{n^{1-\lambda}}{\Delta_i}\right]
    \leq \frac{1}{2(\log n + 1)^2}.
  \end{align*}
  Applying a union bound over all $i_{\max} + 1 \leq \lceil \log_2 n \rceil$ values of $i$,
  the above inequality fails for some $i \in [0, i_{\max}]$ with probability at most
  \begin{align*}
    \frac{\lceil \log_2 n \rceil}{2(\log n + 1)^2} = O\!\left(\frac{1}{\log n}\right).
  \end{align*}
  Hence, with probability at least $1 - O(1/\log n)$,
  $|B^{(i)}| \leq 52(\log n + 1)^3 \cdot n^{1-\lambda} / \Delta_i$
  holds for all $i \in [0, i_{\max}]$ simultaneously.
\end{proof}

\begin{lemma} \label{lma:outgoingEdge}
  \sloppy{
    With probability at least $1 - \tilde{O}(1/n)$, the following statement holds:
    For any $j, k \in [0, i_{\max}]$, any $j$-moderate set $S$, and any $k$-moderate set $U$,
    if $\partial_{\widehat{G}}(S, \Ccal^{(k)} \setminus U)$ contains
    at least $n^{\kappa} (\ln n + 1)$ edges, at least one of them is included in $J$.
  }
\end{lemma}

\begin{proof}
  Fix a tuple $(j, k, S, U)$ for which the premise of the lemma holds, i.e.,
  $|\partial_{\widehat{G}}(S, \Ccal^{(k)} \setminus U)| \geq n^{\kappa} (\ln n + 1)$. Then
  $S \neq \emptyset$ must hold, and since $S$ is $j$-moderate, $|S| \leq n^{2\kappa} \log n / \Delta_j$ yields $\Delta_{\min\{j, k\}} \leq \Delta_j \leq n^{2\kappa} \log n$.
  Since each of these (at least $n^{\kappa}(\ln n + 1)$) edges is sampled with probability at least
  $q_{\min\{j, k\}} = \frac{3 n^{\kappa} \log n}{\Delta_{\min\{j, k\}}}$, we have
  \begin{align*}
    \Pr[\partial_{\widehat{G}}(S, \Ccal^{(k)} \setminus U) \cap J = \emptyset]
     & \leq \left(1 - \frac{3 n^{\kappa} \log n}{\Delta_{\min\{j, k\}}}\right)^{n^{\kappa} (\ln n + 1)} \\
     & \leq \exp\left(- \frac{3 n^{2\kappa} \log n\, (\ln n + 1)}{\Delta_{\min\{j,k\}}} \right),
  \end{align*}
  using $1 - x \leq e^{-x}$.
  (This bound holds trivially even if $q_{\min\{j,k\}} = 1$, since then the left-hand probability is zero.)

  Next, we bound the number of choices of the $j$-moderate set $S$ and the $k$-moderate set $U$. Using
  $\sum_{x \leq K} \binom{N}{x} \leq (eN)^K$ and $\ln(en) = 1 + \ln n$,
  \begin{align*}
    \textrm{\# $j$-moderate sets $S$}
     & \leq \sum_{0 \leq x \leq \lfloor \frac{n^{2\kappa} \log n}{\Delta_j} \rfloor}
    \binom{|\Ccal^{(j)}|}{x}                                                         \\
     & \leq (en)^{\lfloor \frac{n^{2\kappa} \log n}{\Delta_j} \rfloor}               \\
     & \leq \exp\left( \frac{n^{2\kappa} \log n\, (1 + \ln n)}{\Delta_j} \right).
  \end{align*}
  This also applies to the case $n^{2\kappa}\log n / \Delta_j < 1$ (where the only possible choice of $S$
  is the empty set), and the same bound holds for the choice of $U$.
  Hence the number of choices of $(S, U)$ is at most
  $\exp\left(\frac{2 n^{2\kappa} \log n\, (1 + \ln n)}{\Delta_{\min\{j, k\}}}\right)$.

  By the union bound over all choices of $S$ and $U$ (for the fixed pair $(j, k)$),
  the probability that the conclusion fails for some $(S, U)$ is at most
  \begin{align*}
    \exp\left(\frac{2 n^{2\kappa} \log n\, (1 + \ln n)}{\Delta_{\min\{j, k\}}}\right)
    \cdot
    \exp\left(-\frac{3 n^{2\kappa} \log n\, (\ln n + 1)}{\Delta_{\min\{j, k\}}}\right)
    =
    \exp\left(-\frac{n^{2\kappa} \log n\, (1 + \ln n)}{\Delta_{\min\{j, k\}}}\right)
    \leq
    \frac{1}{en},
  \end{align*}
  where the last inequality uses $\Delta_{\min\{j, k\}} \leq n^{2\kappa} \log n$ shown at the beginning of this proof, which makes the exponent
  at least $1 + \ln n$.
  Finally, a union bound over all $(i_{\max} + 1)^2 = O((\log n)^2)$ index pairs
  $(j, k)$ gives total failure
  probability $O((\log n)^2) / (en) = \tilde{O}(1/n)$.
\end{proof}

The lemma below is the key technical observation of our algorithm.

\begin{lemma} \label{lma:closepart}
  Suppose that the statement of Lemma~\ref{lma:outgoingEdge} holds.
  Then $\partial_{\widehat{G}}(P) \subseteq F$ holds for every admissible part $P \in \Pcal$ with high probability.
\end{lemma}

\begin{proof}
  If $\Deg_{\widehat{G}}(P) \leq |V(G[P])| \, n^{\kappa} (\log n + 1)^3$,
  then Step~4 adds all the edges in $I_{\widehat{G}}(P)$ to $F$, and thus the claim is immediate.
  Hence we focus on the case $\Deg_{\widehat{G}}(P) > |V(G[P])| \, n^{\kappa} (\log n + 1)^3$.
  To prove the lemma, it suffices to show
  $|\partial_{\widehat{G}}(P_M, V(\widehat{G}) \setminus P)| \leq n^{\kappa} (\log n + 1)^3$, which implies that $\textsf{OutDetect}(n^{\kappa} (\log n + 1)^3)$ identifies all the
  edges in $\partial_{\widehat{G} - F}(P)$ with high probability, and they are added to $F$.
  Fix any $j, k \in [0, i_{\max}]$. Since $P$ is admissible,
  $|P_M^{(j)}| \leq |P^{(j)}| \leq n^{2\kappa} \log n / \Delta_j$ and
  $|P^{(k)}| \leq n^{2\kappa} \log n / \Delta_k$, i.e., $P_M^{(j)}$ is $j$-moderate and $P^{(k)}$ is $k$-moderate,
  as required by Lemma~\ref{lma:outgoingEdge}. Moreover, since $P_M^{(j)}$ contains no boundary vertex, we
  have $\partial_{\widehat{G}}(P_M^{(j)}, \Ccal^{(k)} \setminus P^{(k)}) \cap J = \emptyset$.
  By Lemma~\ref{lma:outgoingEdge}, $|\partial_{\widehat{G}}(P_M^{(j)}, \Ccal^{(k)} \setminus P^{(k)})| < n^{\kappa} (\ln n + 1)$ holds.
  Summing over all $(i_{\max} + 1)^2 \leq (\log n + 1)^2$ pairs $(j, k)$,
  \begin{align*}
    |\partial_{\widehat{G}}(P_M, V(\widehat{G}) \setminus P)|
     & \leq \sum_{j, k} |\partial_{\widehat{G}}(P_M^{(j)}, \Ccal^{(k)} \setminus P^{(k)})| \\
     & \leq (\log n + 1)^2 \cdot n^{\kappa} (\ln n + 1)
    \leq n^{\kappa} (\log n + 1)^3.
  \end{align*}
  where the last step uses $\ln n + 1 \leq \log n + 1$. This proves the lemma.
\end{proof}

Let $\Pcal_1$ (resp.\ $\Pcal_2$) denote the set of parts of $\Pcal$ that violate condition (C1) (resp.\ (C2)).
Lemma~\ref{lma:closepart} ensures that every admissible part is closed with respect to $F$ at Step~4, and thus
every part left open is non-admissible and hence lies in $\Pcal_1 \cup \Pcal_2$. Moreover, since $J \subseteq F$,
we have $D_{\widehat{G}_F[P]} \leq D_{\widetilde{G}[P]}$ for any part $P$. Thus bounding
$\sum_{P \in \Pcal_1 \cup \Pcal_2} (D_{\widetilde{G}[P]} + 1)$ upper-bounds $\sum_{P \in \Pcal^{\mathrm{open}}_F} (D_{\widehat{G}_F[P]} + 1)$,
the quantity constrained by condition~(\ref{enum:open-part}) of Definition~\ref{dfn:abc-partition}.

\begin{lemma} \label{lma:opendiameter}
  Suppose that the statement of Lemma~\ref{lma:highsize} holds. Then
  \begin{align*}
    \sum_{P \in \Pcal_1 \cup \Pcal_2} \bigl(D_{\widetilde{G}[P]} + 1\bigr)
    = \tilde{O}(n^{1 - \kappa - \lambda} + n^{1 - 2\kappa + \lambda}).
  \end{align*}
\end{lemma}

\begin{proof}
  By Theorem~\ref{thm:diameter-d-decomposition}, $D_{\widetilde{G}[P]} = \tilde{O}(n^{\lambda})$ for every part $P$, and hence $D_{\widetilde{G}[P]} + 1 = \tilde{O}(n^{\lambda})$ as well, since $n^{\lambda} \geq 1$.
  We bound $\sum_{P \in \Pcal_1} (D_{\widetilde{G}[P]} + 1)$ and $\sum_{P \in \Pcal_2} (D_{\widetilde{G}[P]} + 1)$ separately.

  \begin{itemize}
    \item \textbf{Bounding $\sum_{P \in \Pcal_1} (D_{\widetilde{G}[P]} + 1)$:}
          Each $P \in \Pcal_1$ contains at least one vertex of $B \cap \Ccal^H$. This immediately implies $|\Pcal_1| \leq |B \cap \Ccal^H|$.
          By Lemma~\ref{lma:highsize},
          $|B^{(i)}| \leq 52(\log n+1)^3 n^{1-\lambda}/\Delta_i$ holds for all $i$ simultaneously.
          For every $i > t$, we have $\Delta_i \geq \Delta_{t+1} > n^{\kappa+\lambda}\log n$,
          hence $\Delta_{t} > n^{\kappa+\lambda}\log n / 2$.
          Since the geometric series gives $\sum_{i > t} \Delta_i^{-1} \leq \Delta_t^{-1}$, this implies
          \begin{align*}
            |B \cap \Ccal^H|
            \leq \frac{52(\log n+1)^3 \, n^{1-\lambda}}{\Delta_{t}}
            < \frac{104(\log n+1)^3 \, n^{1-\kappa-2\lambda}}{\log n},
          \end{align*}
          and thus
          \begin{align*}
            \sum_{P \in \Pcal_1} (D_{\widetilde{G}[P]} + 1)
            \leq |B \cap \Ccal^H| \cdot \tilde{O}(n^{\lambda})
            \leq \frac{104(\log n+1)^3 \, n^{1-\kappa-2\lambda}}{\log n} \cdot \tilde{O}(n^{\lambda})
            = \tilde{O}(n^{1-\kappa-\lambda}).
          \end{align*}

    \item \textbf{Bounding $\sum_{P \in \Pcal_2} (D_{\widetilde{G}[P]} + 1)$:}
          By Lemma~\ref{lma:balancedcluster}(4), $|\Ccal^{(i)}| \leq 13n\ln n/\Delta_i$.
          For each fixed $i$, the number of parts $P \in \Pcal$ satisfying
          $|P^{(i)}| \geq n^{2\kappa}\log n/\Delta_i$ is at most
          \begin{align*}
            \frac{13n\ln n/\Delta_i}{n^{2\kappa}\log n/\Delta_i}
            = \frac{13n\ln n}{n^{2\kappa}\log n}
            \leq 13n^{1-2\kappa}.
          \end{align*}
          Summing over $i_{\max}+1 \leq \lceil\log_2 n\rceil$ values of $i$:
          \begin{align*}
            \sum_{P \in \Pcal_2} (D_{\widetilde{G}[P]} + 1)
            \leq 13\lceil\log_2 n\rceil \cdot n^{1-2\kappa} \cdot \tilde{O}(n^{\lambda})
            = \tilde{O}(n^{1-2\kappa+\lambda}).
          \end{align*}
  \end{itemize}
  Combining both bounds yields the lemma.
\end{proof}

Finally, we bound the time and message complexity of Algorithm \textsf{Partition}.

\begin{lemma} \label{lma:complexityPartition}
  Let $0 \leq \lambda \leq \kappa \leq 1/2$. With probability $1 - o(1)$, Algorithm \textsf{Partition} runs
  in $\tilde{O}(\alpha + \beta + n^{\kappa})$ rounds and consumes
  $\tilde{O}(\min\{m, n^{1 + \kappa}\})$ messages.
\end{lemma}

\begin{proof}
  The running time is dominated by Steps (3) and (4).
  Step (3) takes $\tilde{O}(n^{\lambda})$ rounds by Theorem~\ref{thm:diameter-d-decomposition}.
  Step (4) requires $\tilde{O}(n^{\lambda})$ rounds for $O(\log n)$ aggregation iterations over
  the spanning tree of each admissible part $P$, and $\tilde{O}(n^{\kappa})$ rounds
  for \textsf{OutDetect} by Theorem~\ref{thm:outdetect}.
  Since $\lambda \leq \kappa$, the total running time is $\tilde{O}(n^{\kappa})$.

  The message complexity of each step is analyzed as follows.
  \begin{itemize}
    \item (Step 1) By Lemma~\ref{lma:balancedcluster}, $\tilde{O}(n)$ messages are consumed.
    \item (Step 2) By a standard Chernoff bound argument,
          $|J| = \tilde{O}(\min\{m, n^{1+\kappa}\})$ with high probability.
    \item (Step 3) By Theorem~\ref{thm:diameter-d-decomposition}, the message cost is
          $O(|J|) = \tilde{O}(\min\{m, n^{1+\kappa}\})$.
    \item (Step 4) Let $\Acal_1$ denote the set of admissible parts $P$
          with $\Deg_{\widehat{G}}(P) \leq |V(G[P])| n^{\kappa} (\log n + 1)^3$,
          and let $\Acal_2$ be the remaining admissible parts.
          Determining the admissibility of $P$ requires $O(\log n)$ iterations of aggregation
          along the spanning tree of $\widetilde{G}[P]$. It costs $\tilde{O}(|V(G[P])|)$ messages per part
          and $\tilde{O}(n)$ in total. This is not a dominant cost.

          For each $P \in \Acal_1$, all edges in $I_{\widehat{G}}(P)$ are added to $F$,
          at a cost of $\Deg_{\widehat{G}}(P) \leq \tilde{O}(|V(G[P])| n^{\kappa})$ messages per part.
          Since $\sum_{P} |V(G[P])| \leq n$, summing over $\Acal_1$ yields $\tilde{O}(n^{1+\kappa})$.

          For $P \in \Acal_2$, we bound the cost of adding the edges incident to $P_B$.
          By Lemma~\ref{lma:highsize},
          $|B^{(i)}| \leq 52(\log n+1)^3 n^{1-\lambda}/\Delta_i$
          with probability $1 - o(1)$ for every $i \in [0, i_{\max}]$.
          By Lemma~\ref{lma:balancedcluster}(2) and (5),
          each cluster $C \in \Ccal^{(i)}$ has $O(\Delta_i)$ vertices each of degree at most $2\Delta_i$,
          and thus $C$ is incident to $O(\Delta_i^2)$ edges in $\widehat{G}$.
          Hence the total number of edges incident to all clusters in $B \cap \Ccal^{(i)}$ is
          $O(|B^{(i)}| \cdot \Delta_i^2) = \tilde{O}(n^{1-\lambda} \Delta_i)$.
          By condition~(C1), an admissible part has no boundary supernode in $\Ccal^H$, i.e., of class $i > t$.
          Hence only levels $i \leq t$ (with $\Delta_i \leq n^{\kappa+\lambda} \log n$) contribute,
          and summing over them,
          \begin{align*}
            \sum_{P \in \Acal_2} \Deg_G(V_G(P_B))
             & = \min\!\left\{2m,\; \sum_{i \leq t} \tilde{O}(n^{1-\lambda} \Delta_i)\right\}                     \\
             & = \min\!\left\{2m,\; \tilde{O}(n^{1-\lambda} \cdot n^{\kappa+\lambda} \log n \cdot \log n)\right\} \\
             & = \tilde{O}(\min\{m, n^{1+\kappa}\}).
          \end{align*}
          For \textsf{OutDetect}, Theorem~\ref{thm:outdetect} gives a cost of
          $\tilde{O}(|V(G[P_M])| \cdot f) = \tilde{O}(|V(G[P])| \cdot n^{\kappa})$ messages per part,
          where $f = n^{\kappa}(\log n+1)^3$.
          Summing over all $P \in \Acal_2$ yields
          $\tilde{O}(|V(G[\Ccal])| \cdot n^{\kappa}) = \tilde{O}(n^{1+\kappa})$.

          Hence, the total message cost of Step 4 is
          \begin{align*}
             & \sum_{P \in \Acal_1} \Deg_{\widehat{G}}(P)
            + \sum_{P \in \Acal_2} \Deg_G(V_G(P_B))
            + \sum_{P \in \Acal_2} \tilde{O}(|P| \cdot n^{\kappa}) \\
             & = \tilde{O}(\min\{m, n^{1+\kappa}\})
            + \tilde{O}(\min\{m, n^{1+\kappa}\})
            + \tilde{O}(|\Ccal| \cdot n^{\kappa})                  \\
             & = \tilde{O}(\min\{m, n^{1+\kappa}\}).
          \end{align*}
  \end{itemize}
  Combining all the steps, we obtain a total message complexity of $\tilde{O}(\min\{m, n^{1+\kappa}\})$.
  The lemma is proved.
\end{proof}

By the lemmas above, we obtain the correctness of \textsf{Partition}.

\begin{theorem} \label{thm:correctnessPartition}
  Let $0 \leq \lambda \leq \kappa \leq 1/2$.
  With probability $1 - o(1)$, the algorithm \textsf{Partition} outputs an $(\alpha, \beta)$-partition with
  \begin{align*}
    \alpha = \tilde{O}(n^{\lambda}), \hspace{5mm} \beta = \tilde{O}(n^{1 - \kappa - \lambda} + n^{1 - 2\kappa + \lambda}).
  \end{align*}
  The running time is $\tilde{O}(n^{\kappa})$ rounds, and the message complexity is $\tilde{O}(\min\{m, n^{1 + \kappa}\})$.
\end{theorem}

Combining with Lemma~\ref{lma:spanner-from-partition}, we obtain Theorem~\ref{thm:main}. The
resulting round complexity is $\tilde{O}(\alpha + \beta + n^{\kappa})$, where the $\tilde{O}(n^{\kappa})$
term is the running time of \textsf{Partition} and $\tilde{O}(\alpha + \beta)$ is that of
Lemma~\ref{lma:spanner-from-partition}.

\section{Concluding Remarks}
\label{sec:conclude}

This paper demonstrated that the quadratic time--message trade-off barrier
$\text{\#rounds} \cdot \text{\#messages} = \tilde{\Omega}(n^2)$,
which all prior algorithms in the \textsf{CONGEST-KT$_1$} model implicitly face,
is \emph{not} a fundamental limit.
Our key contribution is a new $(\alpha, \beta)$-partition framework, together with a
degree-adaptive clustering technique, that yields a message-efficient construction of
$(\alpha, \beta)$-spanners with $\tilde{O}(n^{1+\kappa})$ messages.
For any $0 \leq \lambda \leq \kappa \leq 1/2$, this gives MST algorithms that run in
$\tilde{O}(n^{\lambda}D_G + n^{1-\kappa-\lambda} + n^{1-2\kappa+\lambda} + n^{1/2})$ rounds
and use $\tilde{O}(\min\{m, n^{1+\kappa}\})$ messages.
Setting $(\kappa, \lambda) = (1/3, 1/6)$ gives a nearly time-optimal
$\tilde{O}(n^{1/2} + n^{1/6}D_G)$-round algorithm with only $\tilde{O}(n^{4/3})$ messages,
and more generally the barrier is broken for almost the entire range of the diameter $D_G$.

Our work leaves several natural questions open.
\begin{itemize}
  \item Can the time--message trade-off be improved further, and is the trade-off achieved here optimal?
  \item Can the multiplicative factor $n^{\lambda}$ ($=\alpha$) that multiplies $D_G$ in the round complexity be removed, making the dependence on $D_G$ additive?
  \item Can our approach be applied to obtain message-efficient algorithms for other global problems? Two concrete targets are BFS-tree construction and, more generally, shortest-path-tree construction on weighted graphs, for which no $o(m)$-message algorithm is currently known in the \textsf{CONGEST-KT$_1$} model.
  \item As a more general question, not necessarily tied to our approach, can $o(m)$ message complexity be attained by a \emph{deterministic} algorithm?
\end{itemize}
The last two questions have been answered affirmatively in the less restrictive \textsf{LOCAL} model. Dufoulon, Pandurangan, Robinson, and Scquizzato~\cite{DPRS24} show that in the \textsf{LOCAL-KT$_1$} model, a broad class of problems, including BFS-tree construction, can be solved deterministically using $\tilde{O}(n)$ messages and $\tilde{O}(D_G)$ rounds. Whether analogous guarantees are achievable in the message-limited \textsf{CONGEST-KT$_1$} setting remains open.

\paragraph{Use of AI-based tools.}
In preparing this paper, the authors used AI-based tools solely to assist with language and presentation, namely for English translation, paraphrasing, suggestions on English wording, and checking for grammatical errors and typos. All technical ideas, results, and proofs are entirely the authors' own and were developed without the use of AI. The final manuscript has been carefully reviewed and verified by the authors, who take full responsibility for its content.

\bibliographystyle{plain}
\bibliography{references}

\appendix

\section{Algorithm \textsf{MPX}}
\label{appendix:MPX}

\subsection{Outline}

Let $H = (V(H), E(H))$ be the input graph, and $d > 0$ be the parameter of the algorithm.
The decomposition is obtained via a nearest-center assignment distorted by randomly shifted distances, where the distance
from a vertex $v$ is ``shifted'' by a random value following the exponential distribution with mean $d$.
This procedure admits an equivalent formulation using a shortest-path computation with
an additional super-source vertex. The algorithm works as follows:
\begin{enumerate}
  \item Add a new vertex $s$ to $H$, and connect $s$ to every vertex $v \in V(H)$
        by an edge.
  \item Independently sample $d_v \sim \mathrm{Exp}(d)$ for each $v \in V(H)$, where
        $\mathrm{Exp}(d)$ is the exponential distribution with mean $d$. Then
        assign weight $d_v$ to the edge $(s,v)$. Assign weight one to every edge in $E(H)$.
  \item Compute a shortest path tree from $s$
        in the weighted graph constructed above. For each child $u$ of $s$, the vertices in the
        subtree rooted at $u$ form the part $S(u)$. The algorithm outputs the partition $\Scal$
        consisting of all parts $S(u)$ for all children $u$.
\end{enumerate}

For each vertex $v \in S(u)$, we refer to $u$ as the \emph{center} of $v$, which is denoted by $c(v)$.
We present the key technical lemmas of \textsc{Mpx}, which are derived from the fundamental
statistical properties of the exponential distribution.

\begin{lemma} \label{lma:mpx-max}
  $\max_{v \in V(H)}d_v \leq (\delta + 1) d \ln n$ with probability at least $1- n^{-\delta}$.
\end{lemma}

\begin{lemma} \label{lma:boundaryprob}
  For each $u, v \in V(H)$, define $Y^u_v =\Dist_H(u, v) - d_v$, and let $Y^u_{(i)}$ be the $i$-th smallest value
  in $\{Y^u_v\}_{v \in V(H)}$. Then for any $c \geq 0$, $\Pr[Y^u_{(2)}-Y^u_{(1)} \leq c] < c/d$.
\end{lemma}

The proofs of both lemmas are given in the original paper~\cite{MPX13}. Given an output partition $\Scal$, let $B$ be
the set of all boundary vertices, i.e., vertices with an incident edge crossing two parts.

\begin{lemma}
  Let $B$ be the set of all boundary vertices. For each $u \in V(H)$, $\Pr[u \in B] < 2/d$.
\end{lemma}

\begin{proof}
  If $u$ is a boundary vertex, there exists a neighbor $v$ of $u$
  such that $c(u) \neq c(v)$. Then we have $Y^u_{c(u)} < Y^{u}_{c(v)} \leq Y^{v}_{c(v)} + 1$
  and $Y^v_{c(v)} < Y^{v}_{c(u)} < Y^{u}_{c(u)} + 1$. Combining these inequalities,
  $Y^u_{c(v)} \leq Y^{v}_{c(v)} + 1 < Y^{u}_{c(u)} + 2$ holds.
  Since $Y^{u}_{(2)} \leq Y^{u}_{c(v)}$ holds, we obtain $Y^{u}_{(2)} < Y^{u}_{c(u)} + 2 = Y^{u}_{(1)} + 2$. That is,
  $Y^u_{(2)} - Y^u_{(1)} < 2$. This is a necessary condition for $u \in B$, and can happen with probability at most $2/d$
  by Lemma~\ref{lma:boundaryprob}.
\end{proof}

\section{From Spanner to MST}

In this section, we show how to construct a minimum spanning tree (MST) of the input graph $G$
once an $(\alpha,\beta)$-spanner $G^{\ast}=(V(G),E(G^{\ast}))$ with $M$ edges is given.
The technical idea of our argument follows the approach commonly used in most prior
results~\cite{GP18,GK18}, but the following explanation is slightly modernized using the notions of
partwise aggregation (PA) and the low-congestion shortcut framework.

It is well known that an MST can be constructed by the classical Borůvka's algorithm.
In each phase of Borůvka's algorithm, every connected component selects its minimum-weight
outgoing edge and merges with the corresponding neighboring component.
The main task of each phase is detecting the minimum-weight outgoing edge, which can be abstracted as the
partwise aggregation task.

An instance of PA is associated with a partition $\Pcal$ of $V(G)$, input values $x_v$ assigned to each $v \in V(G)$,
and an associative and commutative operator $\otimes$.
The goal is to compute $\otimes_{v \in P} x_v$ independently for all $P \in \Pcal$ in parallel.
The detection of the minimum-weight outgoing edge in $\partial_G(P)$
can be implemented with $O(\log n)$ invocations of \textsf{OutDetect} with $f = 1$ via a binary-search-like technique~\cite{KKT15}.
Moreover, the simultaneous invocations of \textsf{OutDetect} for all parts in $\Pcal$ are implemented as a single PA instance.
Hence, if PA can be solved in $R_{\mathrm{PA}}$ rounds using $M_{\mathrm{PA}}$ messages,
then an MST can be constructed in $\tilde{O}(R_{\mathrm{PA}})$ rounds and
$\tilde{O}(M_{\mathrm{PA}})$ messages.

Following the message-efficient low-congestion shortcut technique of Haeupler, Hershkowitz, and Wajc~\cite{HHW18},
one can implement PA in $\tilde{O}(n^{1/2} + \alpha D_G + \beta)$ rounds and
$\tilde{O}(M)$ messages (Theorem~1.2 of~\cite{HHW18}).
We first construct a BFS tree $T$ of $G^{\ast}$ rooted at an arbitrary node.
Since $|E(G^{\ast})|=M$, a standard distributed BFS algorithm constructs
$T$ in $O(\alpha D_G + \beta)$ rounds and $O(M)$ messages.
Our algorithm uses $T$ for aggregation in large parts: if $P \in \Pcal$ satisfies $|P| > n^{1/2}$,
then it uses $T$ to complete its aggregation.
Aggregation in small parts $P$ (i.e., $|P| \le n^{1/2}$) is processed locally within each part.

While a straightforward implementation of aggregation in all large parts
requires $n^{3/2}$ messages in the worst case, their technique processes
them using only $\tilde{O}(M)$ messages.\footnote{
  Precisely, our application corresponds to the case $c = O(\sqrt{n})$ and $b = 1$ in Theorem~1.2 of~\cite{HHW18}.}

\end{document}